%% file: main.tex
\newif\iftr\trtrue
\newif\ifPACM\PACMtrue

\documentclass[acmsmall,nonacm]{acmart}

\acmJournal{PACMPL}
\acmVolume{MPI-SWS Tech Report}
\acmNumber{}
\acmArticle{}
\acmYear{2021}
\acmMonth{11}
\startPage{1}

\setcopyright{ccby}

\input{includes}
\input{macros}

\title[Pirouette: Higher-Order Typed Functional Choreographies]{{\small Technical Report MPI-SWS-2021-004 \hfill November 4, 2021}\\Pirouette: Higher-Order Typed Functional Choreographies}
\date{}

\author{Andrew K. Hirsch}
\orcid{0000-0003-2518-614X}
\affiliation{
  \institution{Max Planck Institute for Software Systems}
  \city{Saarland Informatics Campus, Saarbr\"ucken}
  \country{Germany}
}
\email{akhirsch@mpi-sws.org}

\author{Deepak Garg}
\affiliation{
  \institution{Max Planck Institute for Software Systems}
  \city{Saarland Informatics Campus, Saarbr\"ucken}
  \country{Germany}
}
\email{dg@mpi-sws.org}

\begin{document}

\begin{abstract}
  \input{abstract}
\end{abstract}

\maketitle

\input{sections}
\bibliography{main}
\appendix
\section*{Appendices}
\input{appendices}

\end{document}

%% file: includes.tex
\usepackage{amsmath}
\usepackage{amsthm}
\usepackage{bbm}
\usepackage{stmaryrd}
\usepackage{mathtools}
\usepackage{mathpartir}
\usepackage{pl-syntax}
\usepackage{doi}
\usepackage{ifthen}
\usepackage{lo_chor}

%% file: macros.tex
\newcommand{\newcommenter}[3]{%
  \newcommand{#1}[1]{%
    \textcolor{#2}{\small\textsf{[#3: {##1}]}}%
  }%
}
\newcommenter{\akh}{purple}{AKH}
\newcommenter{\dg}{blue}{Deepak}

\theoremstyle{theorem}
\newtheorem{thm}{Theorem}

\newtheorem{cor}{Corollary}
\theoremstyle{definition}

\newtheorem{ex}{Example}


%% file: abstract.tex
We present Pirouette, a language for typed higher-order functional choreographic programming.
Pirouette offers programmers the ability to write a centralized functional program and compile it via \emph{endpoint projection} into programs for each node in a distributed system.
Moreover, Pirouette is defined generically over a (local) language of messages, and lifts guarantees about the message type system to its own.
Message type soundness also guarantees deadlock freedom.
All of our results are verified in Coq.


%% file: sections.tex
\input{intro}
\input{sysmodel}
\input{chor}
\input{eqns}
\input{compiling}
\input{coq}
\input{related}

%% file: intro.tex
\section{Introduction}
\label{sec:introduction}

Higher-order typed functional programming has proven to be a powerful technique for writing single-machine programs.
It allows for strong guarantees through types along with code reuse through higher-order programming.
However, currently, writing \emph{distributed} programs using functional programming requires writing separate code for each node in the distributed system, then using \textsf{send} and \textsf{receive} expressions to transmit data between nodes.
This makes it easy to write code that \emph{deadlocks}, or gets stuck because patterns of \textsf{send}s and \textsf{receive}s do not match.

\emph{Session types}~\citep{ScalasY19,Wadler12,DeYoungCPT12,DardhaGS12,ToninhoCP12,CarboneHY07} offer one solution.
Session types describe the pattern of \textsf{send}s and \textsf{receive}s in a program, allowing a compiler to catch the possibility of deadlock.
However, session types are complicated to work with, and the programmer is still left trying to match up \textsf{send} and \textsf{receive} patterns by hand.

\emph{Choreographic programming}~\citep{Montesi21,Cruz-FilipeM17,Cruz-FilipeM17c,Montesi13,LaneseMZ13,DallaPredaGGLM15} offers another solution.
This is a programming paradigm that writes the distributed program as a single program, ensuring that \textsf{send}s and \textsf{receive}s match by combining \textsf{send} and \textsf{receive} into one construct.
Choreographic languages guarantee \emph{deadlock freedom by design}, so the programmer gets strong guarantees on the communication patterns in their program.
Until now, however, choreographic programming forced the user into a lower-order, imperative, and un(i)typed universe.
This paper presents Pirouette, the first language for choreographic programming which is also higher-order, functional, and typed.

Consider a standard example: the bookseller protocol.
In this protocol, a buyer is looking to buy a book, so they send the book's title to a bookseller.
The bookseller returns its price, and the buyer checks if that price is within their budget.
If they can afford the book, they inform the bookseller who tells them a delivery date for their book; otherwise, they tell the seller they will not buy the book.
After this, the protocol ends.

To write this program with session types, we first create a type describing the interactions from both points of view.
That is, we create a type for the buyer that includes ``send the title to the seller, and then receive back a price'' along with a (dual) type for the seller that includes ``receive the title from the buyer, then send them the price.''
First, we must ensure that the two types match: when the buyer sends to the seller, the seller receives from the buyer.
These types do not describe details like ``if the book is within budget, tell the seller yes.''
Instead, it says ``either send the seller `yes' and receive a delivery date, or send them `no' and end the protocol''.
We then write a program for the buyer and another for the seller, which contain all the nasty details, and a type-checker ensures that these programs have the correct communication patterns.

In contrast, a choreographic programmer would write a single program, and expect the language and compiler to ensure that the resulting programs for the buyer and bookseller are deadlock free.
For instance, such a programmer might write the following:
$$
\Send*{\Buyer}{\textsf{book\_title}}{\Seller}{b}{
  \Send*{\Seller}{\textsf{prices}[b]}{\Buyer}{p}{
    \ChorIf*{\Buyer}{(p < \textsf{budget})}{
      \Sync*{\Buyer}{\LChoice}{\Seller}{
        \Send*{\Seller}{\textsf{get\_delivery\_date}(b)}{\Buyer}{\textsf{d\_date}}{
          \Own{\Buyer}{\textsf{(Some d\_date)}}
        }
      }
    }{
      \Sync*{\Buyer}{\RChoice}{\Seller}{
        \Own{\Buyer}{\textsf{None}}
      }
    }
  }
}
$$
Note that the choreography is a program, not merely the type of another program.
In this program, \Buyer sends a string \textsf{book\_title} to \Seller, who binds that string to the variable~$b$.
\Seller then sends a message to \Buyer, which is the result of looking up the price of the book in a map that \Seller holds.
Note that this message is written in a separate programming language, which defines the map-lookup syntax.
Once \Buyer has the price, they decide whether or not to buy based on their budget.
Either way, they inform \Seller whether they took the left~(\LChoice) or the right~(\RChoice) branch.
In the left branch, \Seller calls some function in the local language to get the delivery date.
(We use the syntax $f(x)$ for function calls in the local syntax to emphasize that these are function calls in the local language, whereas Pirouette uses the syntax $\AppGlobal{F}{X}$ for its function calls.)
\Seller transmits the result of this function call to \Buyer, who returns the date, wrapped so that the entire program returns a value on buyer of type $\ExprOptionType{\Date}$.
In the right branch, \Buyer simply returns \textsf{None}.

Note that communication is always matched up: sends and receives are combined, which guarantees deadlock freedom.
The return value on \Buyer makes this closer to functional programming than previous work on choreographies.
Note that choreographies returning values and having standard simple types (i.e., not session types) are both innovations of this work.

Choreographies excel when more than two parties must interact.
For instance, \citet{CarboneM13} suggest a change to the bookseller protocol similar to the following: \Seller sends the price to two buyers who want to share the purchase of a book, \BI and \BII.
\BII then tells \BI how much they are willing to contribute to the purchase, who then responds to \Seller as before.
We can write a modified choreography as follows:
$$
\Send*{\BI}{\textsf{book\_title}}{\Seller}{b}{
  \Send*{\Seller}{\textsf{prices}[b]}{\BI}{p}{
    \Send*{\Seller}{\textsf{prices}[b]}{\BII}{p}{
      \Send*{\BII}{(p / 2)}{\BI}{\textsf{contrib}}{
        \ChorIf*{\BI}{(p - \textsf{contrib} < \textsf{budget})}{
          \Sync*{\BI}{\LChoice}{\Seller}{
            \Send*{\Seller}{\textsf{get\_delivery\_date}(b)}{\BI}{\textsf{d\_date}}{
              \Own{\BI}{\textsf{(Some d\_date)}}
            }
          }
        }{
          \Sync*{\BI}{\RChoice}{\Seller}{
            \Own{\BI}{\textsf{None}}
          }
        }
      }
    }
  }
}
$$
In order to use session types to describe this communication pattern, we must use \emph{multiparty session types}, a major increase in complexity~\citep{ScalasY19,CarboneHY07}.
However, choreographies handle this without difficulty.

The choreographies explained so far can be expressed in prior work.
This paper introduces \emph{functional choreographies}.
To understand the need for these, notice how we can abstract out a pattern from our two protocols.
A buyer (say~\Buyer) sends a book title to \Seller, who looks up its price.
There is then some process, possibly involving communications, which results in a decision by \Buyer, who then informs \Seller of that decision.
If \Buyer decides to buy the book, they get back a delivery date, otherwise they get nothing.
We can modify the choreography so that the decision is its own function as follows:
$$
\makeatletter \addtocounter{@lochor@numlevels}{1}
\FunGlobal{\textsf{Bookseller}}{F}{\Send*{\Buyer}{\textsf{book\_title}}{\Seller}{b}{
    \DefLocal*{\Buyer}{\textsf{decision}}{
      \AppLocal{\Seller}{F}{(\textsf{prices}[b])}
    }{
      \ChorIf*{\Buyer}{\textsf{decision}}{
        \Sync*{\Buyer}{\LChoice}{\Seller}{
          \Send*{\Seller}{\textsf{get\_delivery\_date}(b)}{\Buyer}{\textsf{the\_date}}{
            \Own{\Buyer}{\textsf{(Some the\_date)}}
          }
        }
      }{
        \Sync*{\Buyer}{\RChoice}{\Seller}{
          \Own{\Buyer}{\textsf{None}}
        }
      }
    }
  }
}
\addtocounter{@lochor@numlevels}{-1}\makeatother
$$
Here, $F$ is a function with a choreographic body which takes a price on \Seller as input and outputs a Boolean on \Buyer.
We can implement either of the previous protocols by changing~$F$:
\begin{mathpar}
\makeatletter \addtocounter{@lochor@numlevels}{1}
  \infer{}{\FunLocal{\Seller}{F}{p}{
    \Send*{\Seller}{p}{\Buyer}{p}{
      \Own{\Buyer}{(p < \textsf{budget})}
    }
  }}\and
\infer{}{\FunLocal{\Seller}{F}{p}{
  \Send*{\Seller}{p}{\Buyer}{p}{
    \Send*{\Seller}{p}{\BII}{p}{
      \Send*{\BII}{(p / 2)}{\Buyer}{\textsf{contrib}}{
        \Own{\Buyer}{(p - \textsf{contrib} < \textsf{budget})}
      }
    }
  }
}}
\addtocounter{@lochor@numlevels}{-1}\makeatother
\end{mathpar}
Moreover, we can give $F$ a type which enforces that it takes a price located on \Seller to a Boolean value on \Buyer, as desired.
We write this type $\LocalFunType{\Seller}{\Price}{\Own{\Buyer}{\Bool}}$.
Then we can give \textsf{Bookseller} the type $\GlobalFunType{(\LocalFunType{\Seller}{\Price}{\Own{\Buyer}{\Bool}})}{\Own{\Buyer}{\Date}}$, indicating that it is a higher-order function which expects an input that abstracts the decision in the bookseller example.
This illustrates the programming convenience gained by combining higher-order functional programming with the choreographic-programming paradigm.
For instance, previous works on choreographies could not treat choreographic functions as data, preventing this form of abstraction~\cite{Cruz-FilipeM17c}.

\paragraph{Contributions}
As stated before, Pirouette is the first (higher-order) choreographic functional programming language.
Pirouette features simple located types of the form $\Own{\Buyer}{\Bool}$ and familiar type constructs like function spaces.
Following the choreography philosophy, Pirouette guarantees deadlock freedom by design, without the use of complicated session types.
Previously, higher-order choreographic programming was only supported through the informal development in Choral~\cite{GiallorenzoMP20}, which is useful despite offering no guarantees.
In the following, we point out some salient technical contributions of our design.

First, Pirouette is \emph{generic} in the language of messages.
Note that, sometimes, locations send non-atomic messages like $p/2$, which can be arbitrary expressions.
In previous work, messages were either from a particular (and very simple) language~\citep{Cruz-FilipeM17} or were assumed to compute to a value in finite time~\citep{CarboneM13}.
Pirouette is defined generically on top of \emph{any} message language (which we refer to as the \emph{local}~language), with very few syntactic constraints.
We further show that a sound type system for the expression language can be lifted to the level of choreographies.

Two key features of choreographies are \emph{out-of-order execution} and \emph{endpoint projection}.
We can think of a choreography as an arbitrary interleaving of communication between programs running at different locations.
However, programs with the same communication pattern can have other interleavings as well.
Choreographies, despite \emph{syntactically} being one arbitrary interleaving are able to \emph{semantically} represent all interleavings by allowing out-of-order execution.
We follow previous work by defining an equivalence relation $\Equiv$ on choreographies to reason about out-of-order execution.
However, unlike previous work, we are able to show that defining our operational semantics by appealing to $\Equiv$ gives a weaker-than-desired semantics.
Pirouette's semantics provides a stronger form of out-of-order execution via a novel combination of a labeled-transition system and a \emph{block set}, which guarantees that out-of-order execution does not violate causality.

Endpoint projection formalizes the intuition of a choreography as representing a collection of programs running at different locations by extracting a program for each location (later, we use the term \emph{control program} for the projected program at each location).
This justifies choreographies as a way of writing distributed programs, and allows us to state and prove that Pirouette programs are deadlock free.
We follow a new design principle for choreographies: \emph{equivalence begets equality}.
That is, equivalent choreographies always project to exactly the same program for each location.
In previous work, equivalence was used more liberally which prevented such a clean theorem.
For instance, \citet{Cruz-FilipeM17} use equivalence to reason about recursion unfolding, so equivalent programs may project to programs where unfolding has and has not been applied.

Moreover, somewhat surprisingly, we show that deadlock freedom is a corollary to the soundness of our type system as well as the soundness and completeness of endpoint projection.
Previous work was able to take advantage of the assumption that messages always produce a value along with the lower-order nature of their choreographies.
Because our choreographies are higher-order and our messages are not assumed to compute, we have to appeal to the soundness of our type system to ensure that our choreographies are always able to take a step.

We have formalized our entire development in Coq and mechanically verified proofs of all of our theorems.
As mentioned by \citet{Cruz-FilipeMP19}, there have been several instances of flaws found in proofs of major theorems in concurrency theory in recent years.
Therefore, in order to show that Pirouette's guarantees about deadlock freedom are trustworthy, we formalize our arguments.
In soon-to-be-published work, a small choreography language has been formalized along with its endpoint-projection operation~\citep{CruzFilipeMP21a,CruzFilipeMP21b}.
However, Pirouette is a much more substantial language than was formalized in that work.

To summarize, we make the following contributions:
\begin{itemize}
\item We introduce Pirouette, the first functional choreography language.
  We present its operational semantics and a (simple) type system.
  The operational semantics allow for out-of-order execution, mimicking the execution of a distributed program.
  These semantics are based on a novel idea of blocking sets -- locations that are blocked on other operations and cannot reduce (Section~\ref{sec:funct-chor}).
\item We describe a general set of constraints on the local (message) language, allowing almost any expression-based language to be used as the local language (Section~\ref{sec:system-model}).
  Type soundness for functional choreographies lifts from type soundness for the local~language (Section~\ref{sec:type-system}).
\item We study equivalence for functional choreographies, which allows for reasoning about out-of-order execution.
  We show that defining out-of-order execution based on this equivalence leads to a weaker-than-desired operational semantics when local reduction is allowed (Section~\ref{sec:equational-reasoning}).
\item We show how endpoint projection extracts programs with explicit send and receive constructs from functional choreographies.
  This translation is sound and complete.
  Moreover, we show that well-typed, projected systems are deadlock free by design, and that this deadlock freedom follows from the soundness of the type system and the soundness and completeness of endpoint projection (Section~\ref{sec:comp-chor}).
\item All our metatheory has been formalized in Coq and all our theorems have been mechanically verified.
  We discuss the particularities of our Coq implementation in Section~\ref{sec:notes-coq-code}.
\end{itemize}

Concurrent, independent work~\citep{CruzFilipeGLMP21} also explores higher-order functional choreographies.
However, they have a very different design philosophy, which leads to a very different technical setup.
For more comparison, see the discussion of related work in Section~\ref{sec:related-work}.

%% file: sysmodel.tex
\section{System Model}
\label{sec:system-model}

We begin by discussing the assumptions we make about the system that Pirouette programs run on.
We have made these assumptions as lightweight as possible.
In particular, we assume that the system consists of a collection of nodes, each of which can run local programs and which can send and receive messages from other nodes.
We also allow local programs to be typed, so that we can lift the type system to the choreography level.

\subsection{Locations}
\label{sec:locations}

We assume a set \Locations of \emph{locations}, which we write as \GenericLoc, \GenericLocI, \GenericLocII, and so on.
These names are treated atomically, so we do not assume any additional operations on locations.
However, we do assume that location equality is decidable, so we can distinguish different locations.

Intuitively, locations refer to nodes in a distributed system.
However, it's worth noting that there is nothing that prevents a node from being a thread, or a process, or any other entity that can run Turing-complete programs and send and receive messages.

\subsection{Communication}
\label{sec:communication}

We assume that every location can communicate with every other location synchronously.
That is, if $\GenericLocI$ sends a message to $\GenericLocII$, then $\GenericLocI$ does not continue until $\GenericLocII$ has received the message, and then $\GenericLocI$ may continue.
Message sending is instantaneous and certain: messages do not get ``lost in the air.''

Nodes should be able to send and receive two kinds of messages: values of local programs (described below) and two special \emph{synchronization messages}, written $\LChoice$ and $\RChoice$.
These will be used in the choreography language to ensure that different locations stay in lock-step with each other.

We also require that each node be able to run a functional \emph{control} program which can send and receive messages, while also running programs in the local language.
We describe the precise requirements in Section~\ref{sec:control-language}, when we have the necessary background.

\subsection{Local Programs}
\label{sec:local-programs}

We assume that every node runs programs in a \emph{local} expression-based language.
Our design treats this language generically, requiring only that it allows certain operations and equations.

Our first requirement is that expressions include \emph{variables}.
We model messages as values, and receipt as binding a value to a variable.
We implement variable binding via substitution, which we write $\Subst{e_1}{\SingleSubst{x}{e_2}}$.
Substitution must satisfy three standard equations:
\begin{itemize}
\item $\Subst{x}{\SingleSubst{x}{e}} = e$
\item $\Subst{e}{\SingleSubst{x}{x}} = e$
\item $\Subst{\Subst{e_1}{\SingleSubst{x}{e_2}}}{\SingleSubst{y}{e_3}} = \Subst{\Subst{e_1}{\SingleSubst{y}{e_3}}}{\SingleSubst{x}{(\Subst{e_2}{\SingleSubst{y}{e_3}})}}$ whenever $x \not\in \FV(e_3)$
\end{itemize}

We require a function $FV(e)$ which returns the set of free variables in $e$.
We require that if $x \not\in \FV(e_1)$, then $\Subst{e_1}{\SingleSubst{x}{e_2}} = e_1$.

We only send values as messages, so we assume a predicate $\ExprValue{e}$ which determines whether $e$ is a value.
We require two special values, $\True$ and $\False$, which we use for branching in choreographies.
We additionally require that all values are closed, in order to allow them to be sent.
To see why, imagine that we send some open expression $e$ from $\GenericLocI$ to $\GenericLocII$.
Since $e$ is open, $e$ contains some free variable $x$, which refers to some data on $\GenericLocI$.
However, when we send $e$ to $\GenericLocII$, this information is lost and $x$ might be captured by a binder in $\GenericLocII$'s program.

Finally, we require that an operational semantics be defined relationally for local expressions.
We write $\ExprStepsTo{e_1}{e_2}$ to denote that $e_1$ steps to $e_2$ in the operational semantics.
The only requirement on this semantics is that values do not take steps: if $\ExprValue{v}$, then $\NotExprStepsTo{v}{e}$ for any $e$.

\subsubsection{Examples}
\label{sec:loc-lang-examples}

Simply-typed functional languages easily satisfy the requirements above.
However, any expression-based language can be used, not only ones which define functions.
We discuss two examples here.

\begin{ex}[Call-by-Value $\lambda$-Calculus]
  \label{ex:cbv-lambda}
 The call-by-value $\lambda$-calculus, extended with recursive functions, Boolean values and if-then-else expression, almost fits our requirements.
  However, we must restrict values to be closed.
\end{ex}

\begin{ex}[A Natural-Number Language]
  \label{ex:nl}
  We provide a language with the following syntax:
  \begin{syntax}
    \category[Expressions]{e} \alternative{x} \alternative{0} \alternative{S\,e} \alternative{\True} \alternative{\False}
  \end{syntax}
  Intuitively, $0$ stands for $0$ as a natural number, $S$ is the successor operation on natural numbers, and \True and \False stand for Boolean truth and falsity, respectively.
  Since there are no binders, substitution is easy to define.
  Any closed term is a value.

  This is similar to the language of messages in~\citet{Cruz-FilipeM17}, but modified to fit our requirements.
  In particular, we (a) allow more than one variable, which is treated via substitution instead of as a reference to state, and (b) add the \True and \False terms.

  We include this language to demonstrate that our requirements do not force the choice of $\lambda$-calculus as an expression language.
  Indeed, we will see later that we can equip this language with a type system which results in a sound choreographic type system.
\end{ex}

\subsection{Typed Local Programs}
\label{sec:typed-local-program}

The guarantees we provide for Pirouette depend on the guarantees provided by the local~language's type system.
If we do not consider the local type~system, then we are able to provide a sound and complete translation to a language with explicit send and receive constructs.
If the local type~system guarantees preservation, but not progress (as in a unityped system), then we are also able to prove preservation of the choreographic type system.
Finally, soundness of the local type~system implies not only type soundness for Pirouette, but also deadlock freedom.

We allow the types for the local language to be any language of simple types.
We require that \True and \False have the same type, which we refer to as \Bool.
Note, however, that this type may not be called \Bool, as we will see in Example~\ref{ex:unityped-lambda}.

We assume that the type system can be presented in terms of a judgment $\ExprTyping{\Gamma}{e}{\ExprTypeTau}$, where $\Gamma$ is a sequence of variable-type pairs written $\ExprTypingPair{x}{\ExprTypeTau}$.
We additionally require that typing be unique: if $\ExprTyping{\Gamma}{e}{\ExprTypeTauI}$ and $\ExprTyping{\Gamma}{e}{\ExprTypeTauII}$, then $\ExprTypeTauI = \ExprTypeTauII$.
We only use this requirement to show that types play well with equational reasoning (Theorem~\ref{thm:equiv-types}).
While this requirement is unusual, we note that it is usually satisfied by simple-type systems.
We conjecture that this requirement could be removed with a small change to our choreographic language; we return to this point in Section~\ref{sec:equational-reasoning}.

The other requirements can all be framed as admissible rules; those rules can be found in Figure~\ref{fig:local-type-rules}.
Since these rules should be admissible, we do \emph{not} require that these are actual rules in the type system; merely that we can use them to build typing proofs.
This will be important when building typing proofs for choreographies on top of typing proofs for the local language.

\begin{figure}
  \begin{mathpar}
    \infer[Var]{ }{\ExprTyping{\Gamma, \ExprTypingPair{x}{\ExprTypeTau}}{x}{\ExprTypeTau}} \and
    \infer[True]{ }{\ExprTyping{\Gamma}{\True}{\Bool}} \and \infer[False]{ }{\ExprTyping{\Gamma}{\False}{\Bool}} \\
    \infer[Exchange]{\ExprTyping{\Gamma, \ExprTypingPair{x}{\ExprTypeTauI}, \ExprTypingPair{y}{\ExprTypeTauII}}{e}{\ExprTypeTauIII}}
    {\ExprTyping{\Gamma, \ExprTypingPair{y}{\ExprTypeTauII}, \ExprTypingPair{x}{\ExprTypeTauI}}{e}{\ExprTypeTauIII}} \and
    \infer[Weakening]{\ExprTyping{\Gamma}{e}{\ExprTypeTauI}}{\ExprTyping{\Gamma, \ExprTypingPair{y}{\ExprTypeTauII}}{e}{\ExprTypeTauI}}\\
    \infer[Strengthening]{\ExprTyping{\Gamma, \ExprTypingPair{x}{\ExprTypeTau}}{e}{\ExprTypeTau}\\ x \notin \FV(e)}{\ExprTyping{\Gamma}{e}{\ExprTypeTau}} \and
    \infer[Substitution]{\ExprTyping{\Gamma}{e_1}{\ExprTypeTauI}\\ \ExprTyping{\Gamma, \ExprTypingPair{x}{\ExprTypeTauI}}{e_2}{\ExprTypeTauII}}
    {\ExprTyping{\Gamma}{\Subst{e_2}{\SingleSubst{x}{e_1}}}{\ExprTypeTauII}}
  \end{mathpar}
  \caption[Required-Admissible Rules]{Required-Admissible Rules for Local Type Systems}
  \label{fig:local-type-rules}
\end{figure}

First, we require that variables be typed according to the context, and that \True and \False be typed in the \Bool type.
We also require that the standard structural rules of \textsc{Exchange} and \textsc{Weakening} be allowed.
The \textsc{Strengthening} rule is unusually-presented, but is common in most type systems: non-free variables can be safely removed from the typing context.
Finally, we require the standard property of \textsc{Substitution}: substituting a well-typed variable into a well-typed expression yields a well-typed expression.

\paragraph{Sound Type Systems}
\label{sec:sound-type-systems}
We say that a local type system is \emph{sound} if it additionally satisfies the following three requirements:
\begin{itemize}
\item (\textsc{Boolean Invertibility}) The type~\Bool is invertible for values: if $v$ is a value and $\ExprTyping{ }{v}{\Bool}$ then either $v = \True$ or $v = \False$.
\item (\textsc{Preservation}) If $\ExprStepsTo{e_1}{e_2}$ and $\ExprTyping{\Gamma}{e_1}{\ExprTypeTau}$, then $\ExprTyping{\Gamma}{e_2}{\ExprTypeTau}$.
\item (\textsc{Progress}) If $e_1$ is closed and $\ExprTyping{ }{e_1}{\ExprTypeTau}$, then either $e_1$ is a value or there is an $e_2$ such that $\ExprStepsTo{e_1}{e_2}$.
\end{itemize}

\subsubsection{Examples}
\label{sec:typed-local-examples}
Each of the languages that served as examples above can be given type systems that satisfy the requirements above.
In fact, we can give the $\lambda$-calculus two type systems, though one is not sound.

\begin{ex}[Simply-Typed $\lambda$-Calculus]
  \label{ex:stlc}
  This is the paradigmatic example of a typed local language.
  None of the rules in Figure~\ref{fig:local-type-rules} are difficult; most of the proofs are completely standard.
  Moreover, it is sound: the progress and preservation proofs are standard, and the invertability proof is easy.
\end{ex}

\begin{ex}[Typed Natural-Numbers]
  \label{ex:typed-nats}
  We use two types: \Int and \Bool.
  Then, we have the following rules:
  \begin{mathpar}
    \infer{\ExprTypingPair{x}{\ExprTypeTau} \in \Gamma}{\ExprTyping{\Gamma}{x}{\tau}} \and
    \infer{ }{\ExprTyping{\Gamma}{0}{\Int}} \and \infer{\ExprTyping{\Gamma}{t}{\Int}}{\ExprTyping{\Gamma}{S\,t}{\Int}} \and
    \infer{ }{\ExprTyping{\Gamma}{\True}{\Bool}} \and \infer{ }{\ExprTyping{\Gamma}{\False}{\Bool}}
  \end{mathpar}
  This makes the rules in Figure~\ref{fig:local-type-rules} easy to prove, and soundness is trivial.
\end{ex}

\begin{ex}[Unityped $\lambda$-Calculus]
  \label{ex:unityped-lambda}
  Using the idea of ``untyped is unityped'', we develop a type system for $\lambda$-calculus with only one type, $\ast$, and only one typing rule: $$\infer[Trivial]{ }{\ExprTyping{\Gamma}{e}{\ast}}$$
  Here, we satisfy the requirement to have a type for Boolean values by setting $\Bool = \ast$.
  The rules in Figure~\ref{fig:local-type-rules} are trivial, as is \textsc{Preservation}.
  However, this system is \emph{not} sound: $\Bool$ is not invertible, and uni-typed $\lambda$-calculus does not satisfy progress (programs may get stuck).
\end{ex}


%% file: chor.tex
\section{Functional Choreographies}
\label{sec:funct-chor}

We introduce Pirouette, our language for writing distributed programs in a functional, choreographic style.
We guarantee deadlock freedom along with the standard static guarantees of simple-type systems.
Moreover, we allow higher-order computations by allowing choreographies to be passed to other choreographies as inputs.

\begin{figure}
  \begin{syntax}
    \abstractCategory[Locations]{\GenericLoc \in \Locations}
    \category[Synchronization Labels]{\GenChoice} \alternative{\LChoice} \alternative{\RChoice}
    \category[Choreographies]{C} \alternative{X} \alternative{\Own{\GenericLoc}{e}} \alternative{\Send{\GenericLocI}{e}{\GenericLocII}{x}{C}}
    \alternativeLine{\ChorIf{\GenericLoc}{e}{C_1}{C_2}} \alternative{\Sync{\GenericLocI}{\GenChoice}{\GenericLocII}{C}}
    \alternativeLine{\DefLocal{\GenericLoc}{x}{C_1}{C_2}} \alternative{\FunLocal{\GenericLoc}{F}{x}{C}} \alternative{\FunGlobal{F}{X}{C}}
    \alternativeLine{\AppLocal{\GenericLoc}{C}{e}} \alternative{\AppGlobal{C_1}{C_2}}
  \end{syntax}
  \caption{Functional Choreographies Syntax}
  \label{fig:chor-syn}
\end{figure}

The syntax of Pirouette can be found in Figure~\ref{fig:chor-syn}.
Note that choreographies contain two types of variables: choreography variables, which stand for the result of a distributed computation (i.e., a choreography), and local variables which are the variables of the local language seen in Section~\ref{sec:system-model}.
We write choreography variables in upper case ($X$, $Y$, etc.) and choreography variables in lower case ($x$, $y$, etc.).

Every location has its own namespace of local variables, so $\Own{\GenericLocI}{x} \neq \Own{\GenericLocII}{x}$.
This is reflected in the defintion of substitution: substitution of choreography variables has the standard form $\Subst{C_1}{\SingleSubst{X}{C_2}}$,
while substitution of local variables has the form $\Subst{C}{\SingleChorExprSubst{\GenericLoc}{x}{e}}$.
The first notion of substitution is standard.
The second walks through a term looking for a subterm of the form $\Own{\GenericLoc}{e'}$, and then replaces $e'$ with $\Subst{e'}{\SingleSubst{x}{e}}$.
(Note that the location name~$\GenericLoc$ must be the same as the one in the substitution!)

We write $\Own{\GenericLoc}{e}$ to represent the choreography that returns the result of running $e$ on $\GenericLoc$.
To use this value in a future choreography, we use the syntax $\DefLocal{\GenericLoc}{x}{C_1}{C_2}$.
This runs $C_1$ until it returns a value on \GenericLoc, and then binds that value to $x$ under \GenericLoc in $C_2$.

We write $\Send{\GenericLocI}{e}{\GenericLocII}{x}{C}$ to represent $\GenericLocI$ evaluating $e$ and then sending the resulting value~$v$ to $\GenericLocII$.
The variable~$\Own{\GenericLocII}{x}$ is bound to $v$ in the continuation choreography~$C$, representing $\GenericLocII$'s receipt of the message.

The distributed program can branch based on the result of a test on a local machine.
We write this $\ChorIf{\GenericLoc}{e}{C_1}{C_2}$.
However, this can quickly break causality.
To see why, consider the program
$
\ChorIf{\GenericLocI}{e}
{
  \Own{\GenericLocII}{3}
}
{
  \Own{\GenericLocII}{4}
}
$.
In this program, $\GenericLocII$ behaves differently depending on the behavior of a test performed on $\GenericLocI$, but $\GenericLocII$ has never been told about the result of that test!
In order to fix this problem, we require that $\GenericLocI$ inform $\GenericLocII$ of which branch was taken before $\GenericLocII$ can behave differently in the two branches.
We do this using the syntax $\Sync{\GenericLocI}{\GenChoice}{\GenericLocII}{C}$.
Here \GenChoice can either be \LChoice or \RChoice, where \LChoice represents taking the then branch, while \RChoice represents taking the else branch.
Thus, we can safely write the following program:
$$
\ChorIf*{\GenericLocI}{e}
{
  \Sync{\GenericLocI}{\LChoice}{\GenericLocII}{\Own{\GenericLocII}{3}}
}
{
  \Sync{\GenericLocI}{\RChoice}{\GenericLocII}{\Own{\GenericLocII}{4}}
}
$$

Note that this is only required because \emph{\GenericLocII's behavior differs} in the two branches.
If \GenericLocII behaves the same, no synchronization is required.
Thus, the following program is okay:
$$\ChorIf*{\GenericLocI}{e}
{
  \Send*{\GenericLocIII}{e'}{\GenericLocII}{x}
  {\Sync{\GenericLocI}{\LChoice}{\GenericLocII}{\Own{\GenericLocII}{x + 2}}}
}
{
  \Send*{\GenericLocIII}{e'}{\GenericLocII}{x}{
    \Sync{\GenericLocI}{\RChoice}{\GenericLocII}{\Own{\GenericLocII}{x + 3}}}
}
$$
Because it behaves the same in both branches, \GenericLocIII never needs to be informed about which branch is taken, even though it appears in the branches.
Moreover, \GenericLocII only needs to be informed about which branch is taken after it receives its message from \GenericLocIII.

There are two types of functions available in Pirouette, local and global functions.
Local functions expect a local value as input, stored on some particular node, whereas global functions expect a choreography as input.
Both types of functions may be recursive.
We write $\FunLocal{\GenericLoc}{F}{x}{C}$ for the (recursive) function named $F$ which expects a local value on \GenericLoc as input, and has body $C$.
We write $\AppLocal{\GenericLoc}{C}{e}$ for the application of function $C$ to input $e$ which is stored on \GenericLoc.
For the global function named $F$ which takes an input named $X$ and has body $C$, we write $\FunGlobal{F}{X}{C}$.
As is traditional, we write $\AppGlobal{C_1}{C_2}$ for the application of function $C_1$ to $C_2$.

We adopt a call-by-value semantics, so we evaluate inputs to values before applying functions.
As described in Section~\ref{sec:system-model}, local values are defined by the local language and are always closed.
Choreography values are programs of any of the forms: (a)~$\Own{\GenericLoc}{v}$ (where $v$ is a local value), (b)~$\FunLocal{\GenericLoc}{F}{x}{C}$ (where the only free variables in $C$ are $F$ and $\Own{\GenericLoc}{x}$), or (c)~$\FunGlobal{F}{X}{C}$ (where the only free variables in $C$ are $F$ and $X$).

We define an operation $\LocationNames(C)$ which collects all of the location names in $C$ in a set.
We write $\FCV(C)$ for the free choreography variables in $C$, and $\FEV(C)$ for the collection of free expression variables in $C$, tagged with the locations that own them.
We write $\FEV[\GenericLoc](C)$ for the collection of expression variables free under \GenericLoc in C.

\subsection{Operational Semantics}
\label{sec:oper-semant}

Intuitively, if two locations~\GenericLocI and \GenericLocII both take actions, they should be able to do this in either order.
For instance, consider the Pirouette program
$
\Send{\GenericLocI}{2+3}{\GenericLocII}{x}{
  \Send{\GenericLocIII}{3*4}{\GenericLocII}{y}{C}}
$.
Here, \GenericLocI and \GenericLocIII are both working on computations whose results they expect to send to \GenericLocII.
Since these are different locations, both should be able to work on their programs at the same time.
Thus in the semantics we should be able to reduce this program to either
$
\Send{\GenericLocI}{5}{\GenericLocII}{x}{\Send{\GenericLocIII}{3*4}{\GenericLocII}{y}{C}}$ or $\Send{\GenericLocI}{2+3}{\GenericLocII}{x}{\Send{\GenericLocIII}{12}{\GenericLocII}{y}{C}}$,
performing the local reductions in either order.
However, \GenericLocII is just one location, so it cannot listen for two messages at once.
Thus, even if we reduce the above program to
$
\Send{\GenericLocI}{5}{\GenericLocII}{x}{
  \Send{\GenericLocIII}{12}{\GenericLocII}{y}{C}}
$
we are forced to reduce the send from \GenericLocI before the send from \GenericLocIII, since the second is waiting on \GenericLocII.

In order to allow for this behavior, we keep track of a set of \emph{blocked} locations in our operational semantics.
Intuitively, blocked locations cannot take a step.
By keeping track of what locations are blocked, we can allow out-of-order-execution among non-blocked locations.

However, this is not quite enough to get the behavior we want.
Consider this program, which represents a system where \GenericLocI branches on $e$ and then does nothing, while \GenericLocII returns the result of~$3 + 5$ independent of \GenericLocI's choice: \mbox{$\ChorIf{\GenericLocI}{e}{\Own{\GenericLocII}{3+5}}{\Own{\GenericLocII}{3+5}}$}.
Here, \GenericLocII ought to be able to make progress, reducing this program to $\ChorIf{\GenericLocI}{e}{\Own{\GenericLocII}{8}}{\Own{\GenericLocII}{8}}$.
Note that this progress was atomic.
Thus, it would be illegitimate to reduce different programs in the two branches.
For instance, reducing
$\ChorIf{\GenericLocI}{e}{\Own{\GenericLocII}{3+5}}{\Own{\GenericLocII}{3*4}}$ to $\ChorIf{\GenericLocI}{e}{\Own{\GenericLocII}{8}}{\Own{\GenericLocII}{12}}$
requires reducing different programs in the two branches, whereas, in reality, only one of those two programs should be reduced and that program should be chosen by \GenericLocI.

In order to prevent this type of reduction, we track \emph{what} reduction is happening, resulting in a labeled transition system.
We refer to labels in Pirouette as ``redices,'' and to a single label as a ``redex.''
In the second-to-last example above, we note that in each branch, we reduce a 3+5 to 8 on \GenericLocII.
Since we are doing the same reduction in each branch, we can reduce the whole program.

\begin{figure}
  \begin{mathpar}
    \infer[SendE]{\GenericLocI \notin B\\ \GenericLocI \neq \GenericLocII\\ \ExprStepsTo{e_1}{e_2}}{
      \ChorStepsTo{\RSendE{\GenericLocI}{e_1}{e_2}{\GenericLocII}}{B}{\Send{\GenericLocI}{e_1}{\GenericLocII}{x}{C}}{\Send{\GenericLocI}{e_2}{\GenericLocII}{x}{C}}} \and
    \infer[SendI]{\ChorStepsTo{R}{B \mathrel{\cup} \{\GenericLocI, \GenericLocII\}}{C_1}{C_2}}{
      \ChorStepsTo{R}{B}{\Send{\GenericLocI}{e}{\GenericLocII}{x}{C_1}}{\Send{\GenericLocI}{e}{\GenericLocII}{x}{C_2}}}\and
    \infer[SendV]{\GenericLocI \notin B\\ \GenericLocII \notin B\\ \ExprValue{v}\\ \GenericLocI \neq \GenericLocII}{
      \ChorStepsTo{\RSendV{\GenericLocI}{v}{\GenericLocII}}{B}{\Send{\GenericLocI}{v}{\GenericLocII}{x}{C}}{\Subst{C}{\SingleChorExprSubst{\GenericLocII}{x}{v}}}}\\
    \infer[IfI]{
      \ChorStepsTo{R}{B \cup \{\GenericLoc\}}{C_1}{C_1'}\\
      \ChorStepsTo{R}{B \cup \{\GenericLoc\}}{C_2}{C_2'}
    }
    {
      \ChorStepsTo{R}{B}{\ChorIf{\GenericLoc}{e}{C_1}{C_2}}{\ChorIf{\GenericLoc}{e}{C_1'}{C_2'}}
    } \\
    \infer[SyncI]{
      \ChorStepsTo{R}{B \cup \{\GenericLoc, \GenericLocII\}}{C_1}{C_2}
    }{
      \ChorStepsTo{R}{B}{\Sync{\GenericLocI}{\GenChoice}{\GenericLocII}{C_1}}{\Sync{\GenericLocI}{\GenChoice}{\GenericLocII}{C_2}}
    } \and
    \infer[Sync]{
      \GenericLocI \notin B\\ \GenericLocII \notin B\\
      \GenericLocI \neq \GenericLocII
    }{
      \ChorStepsTo{\RSync{\GenericLocI}{\GenChoice}{\GenericLocII}}{B}{\Sync{\GenericLocI}{\GenChoice}{\GenericLocII}{C}}{C}
    }\\
    \infer[AppGlobalFun]{
      \ChorStepsTo{R}{B}{C_1}{C_1'}
    }{
      \ChorStepsTo{\RFun(R)}{B}{\AppGlobal{C_1}{C_2}}{\AppGlobal{C_1'}{C_2}}
    }\and
    \infer[AppGlobalArg]{
      \ChorStepsTo{R}{B}{C_2}{C_2'}
    }{
      \ChorStepsTo{\RArg(R)}{B}{\AppGlobal{C_1}{C_2}}{\AppGlobal{C_1}{C_2'}}
    } \and
    \infer[AppGlobal]{
      \ChorValue(V)
    }{
      \ChorStepsTo{\RAppGlobal}{\emptyset}{\AppGlobal{(\FunGlobal{F}{X}{C})}{V}}{\Subst{C}{\SingleSubst{X}{V}, \SingleSubst{F}{\FunGlobal{F}{X}{C}}}}
    }
  \end{mathpar}
  \caption{Selected Choreography Operational Semantics}
  \label{fig:chor-sem}
\end{figure}

Selected rules from the operational semantics can be found in Figure~\ref{fig:chor-sem}.
The full rules, along with the syntax of redices, can be found in \iftr{}Appendix~\ref{sec:block-set-semantics}\else{}the accompanying technical report\fi{}.

Each rule has the form $\ChorStepsTo{R}{B}{C_1}{C_2}$, where $R$ is a redex, $B$ is a set of locations, and $C_1$ and $C_2$ are choreographies.
Intuitively, this says that $C_1$ can reduce to $C_2$ using redex $R$ even if every location in $B$ is blocked.
With this interpretation, the rule~\textsc{SendE} allows a location to reduce a message to a value before sending it.
The notation $\RSendE{\GenericLocI}{e_1}{e_2}{\GenericLocII}$ is the redex for this reduction rule.
Note that we check both that $\GenericLocI \notin B$ and that $\GenericLocI \neq \GenericLocII$ in the premise of this rule.
The first check ensures that $\GenericLocI$ is not blocked, since the reduction happens at that location.
The second check reflects the fact that sends from a node to itself is not meaningful.

The rule~\textsc{SendI} allows reductions under a send construct, but only when the reduction can take place with both the sender and the receiver blocked.
This formalizes the intuition that while some locations are waiting, other locations can take actions.
Sending is formalized with the rule~\textsc{SendV}.
This removes a send entirely.
Note that in the substitution, \GenericLocII's variable is substituted with the message.
This formalizes the intuition that sends are modeled by binding the message to a variable in the receiver's program.

The \textsc{IfI} rule is the only rule that makes use of redices.
By ensuring that the same redex is reduced on each branch of the if, we make sure that the only reductions that can be made are those which are invariant under which branch is taken.
Note that we again ensure that \GenericLocI is not taking an action in this reduction step, since it is currently occupied with the if itself.

The rules \textsc{SyncI} and \textsc{Sync} demonstrate the rather subtle effect that \textsc{Sync} has on the semantics.
The first rule demonstrates that synchronization induces blocking on the relevant locations.
However, \textsc{Sync} shows that this is the only effect that it has on the semantics of the choreography.

Finally, the rules \textsc{AppGlobalFun}, \textsc{AppGlobalArg}, and \textsc{AppGlobal} demonstrate how functions are treated.
Both the function and its argument can be evaluated in any order, but the redex is tagged with which choice is made.
This ensures that if statements don't reduce a function in one branch and its argument in the other.
Finally, the \textsc{AppGlobal} rule shows that the semantics of a function is given as standard, with a parallel substitution.
However, we require that there be no blockers.
This comes from the fact that all of the locations work together to reduce functions at the same time.

\subsection{Type System}
\label{sec:type-system}

\begin{figure}
  \begin{syntax}
    \category[Choreography Types]{\ChorTypeTau}
    \alternative{\OwnedType{\GenericLocI}{\ExprTypeTau}}
    \alternative{\LocalFunType{\GenericLoc}{\ExprTypeTauI}{\ChorTypeTauII}}
    \alternative{\GlobalFunType{\ChorTypeTauI}{\ChorTypeTauII}}
    \category[Local Contexts]{\Gamma} \alternative{\cdot} \alternative{\Gamma, \ChorTypingTriple{\GenericLoc}{x}{\ExprTypeTau}}
    \category[Choreography Contexts]{\Delta} \alternative{\cdot} \alternative{\Delta, \ChorTypingPair{X}{\ChorTypeTau}}      
  \end{syntax}

  \begin{mathpar}
    \infer[Send]{\ExprTyping{\ProjectCtxt{\Gamma}{\GenericLocI}}{e}{\ExprTypeTauI}\\
      \ChorTyping{\Gamma, \ChorTypingTriple{\GenericLocII}{x}{\ExprTypeTauI}}{\Delta}{C}{\ChorTypeTauII}\\
      \GenericLoc \neq \GenericLocII
    }
    {
      \ChorTyping{\Gamma}{\Delta}{\Send{\GenericLocI}{e}{\GenericLocII}{x}{C}}{\ChorTypeTauII}
    } \and
  \infer[DefLocal]{
    \ChorTyping{\Gamma}{\Delta}{C_1}{\OwnedType{\GenericLoc}{\ExprTypeTauI}}\\
    \ChorTyping{\Gamma, \ChorTypingTriple{\GenericLoc}{x}{\ExprTypeTauI}}{\Delta}{C_2}{\ChorTypeTauII}}
  {\ChorTyping{\Gamma}{\Delta}{\DefLocal{\GenericLoc}{x}{C_1}{C_2}}{\ChorTypeTauII}}\and
  \infer[FunLocal]{
    \ChorTyping{\Gamma, \ChorTypingTriple{\GenericLoc}{x}{\ExprTypeTauI}}{\Delta, \ChorTypingPair{F}{\LocalFunType{\GenericLoc}{\ExprTypeTauI}{\ChorTypeTauII}}}{C}{\ChorTypeTauII}
  }{
    \ChorTyping{\Gamma}{\Delta}{\FunLocal{\GenericLoc}{F}{x}{C}}{\LocalFunType{\GenericLoc}{\ExprTypeTauI}{\ChorTypeTauII}}
  }\and
  \infer[FunGlobal]{
    \ChorTyping{\Gamma}{\Delta, \ChorTypingPair{F}{\GlobalFunType{\ChorTypeTauI}{\ChorTypeTauII}}, \ChorTypingPair{X}{\ChorTypeTauI}}{C}{\ChorTypeTauII}
  }{
    \ChorTyping{\Gamma}{\Delta}{\FunGlobal{F}{X}{C}}{\GlobalFunType{\ChorTypeTauI}{\ChorTypeTauII}}
  }
  \end{mathpar}
  
  \caption[Pirouette Types]{Pirouette Types (Syntax and Selected Rules)}
  \label{fig:chor-types}
\end{figure}

Assuming that the local expression language has a type system as described in Section~\ref{sec:typed-local-program}, we can develop a language of types for Pirouette.
The syntax of Pirouette types can be found in Figure~\ref{fig:chor-types}, along with selected typing rules.
The full typing rules are in \iftr{}Appendix~\ref{sec:full-chor-types}\else{}the accompanying technical report~\cite{HirschG21}\fi{}.

There are three categories of Pirouette types, corresponding to the three categories of choreographic values.
The first category is a local type at some location, which we write $\OwnedType{\GenericLoc}{\ExprTypeTau}$.
This is the type given to values of the form $\Own{\GenericLoc}{v}$.
Then there is the type of local functions, which we write $\LocalFunType{\GenericLoc}{\ExprTypeTauI}{\ChorTypeTauII}$.
Here, $\ExprTypeTauI$ is an expression type, corresponding to the type of the input to the function, which is expected to be located on $\GenericLoc$.
The function then returns a~$\ChorTypeTauII$, which is a choreography type.
Finally, there is the type of global functions which take an arbitrary choreographic type as an input.
Thus, we write $\GlobalFunType{\ChorTypeTauI}{\ChorTypeTauII}$, where $\ChorTypeTauI$ and $\ChorTypeTauII$ are both Pirouette types.

The choreographic type judgment is of the form $\ChorTyping{\Gamma}{\Delta}{C}{\ChorTypeTau}$ where $\Gamma$ is a local context, $\Delta$ is a global context, $C$ is a choreography, and $\ChorTypeTau$ is a choreographic type.
Global contexts are normal typing contexts relating choreography variables to choreography types.
However, local contexts have to relate variables to types, but must also keep track of the location of the variable.
Since each location is its own namespace, we relate variables paired with locations to local types (see the syntax of $\Gamma$ in Figure~\ref{fig:chor-types}).

For any location~$\GenericLoc$, we can recursively \emph{project} the typing context of the location from a local context~$\Gamma$ as follows:
$$
\ProjectCtxt{\Gamma}{\GenericLocI} =
\left\{
  \begin{array}{ll}
    \cdot & \text{if}~\Gamma = \cdot\\
    \ProjectCtxt{\Gamma'}{\GenericLocI}, \ExprTypingPair{x}{\ExprTypeTau} & \text{if}~\Gamma = \Gamma,\ChorTypingTriple{\GenericLocI}{x}{\ExprTypeTau}\\
    \ProjectCtxt{\Gamma'}{\GenericLocI} & \text{if}~\Gamma = \Gamma,\ChorTypingTriple{\GenericLocII}{x}{\ExprTypeTau}~\text{where}~\GenericLocI \neq \GenericLocII\\
  \end{array}
\right.
$$
Intuitively, this gives the context for the namespace $\GenericLoc$.

We can see projection in action in the \textsc{Send} rule.
Here, we check that $e$ has the local type~$\ExprTypeTau$ at the location $\GenericLoc$.
We then check the remainder of the choreography under the assumption that $x$ has type $\ExprTypeTau$ \emph{at $\GenericLocII$}, since after the send $x$ will be bound to the result of evaluating $e$.
Note that this works because values are closed; otherwise, $v$ might not typecheck in \GenericLocII's namespace.

The rule \textsc{DefLocal} shows how the program $\DefLocal{\GenericLoc}{x}{C_1}{C_2}$ acts as an elimination rule for the type~$\OwnedType{\GenericLoc}{\ExprTypeTau}$.
We ensure that $C_1$ has the type $\OwnedType{\GenericLoc}{\ExprTypeTau}$, and then we bind $x$ to $\ExprTypeTau$ locally in $C_2$.
Finally, \textsc{FunLocal} and \textsc{FunGlobal} produce local and global function types, respectively.
Both also bind the function name to the appropriate function type, allowing for recursive function definitions.

The choreographic type system enjoys progress and preservation if the local type system is sound.
However, we can get more-precise guarantees, which we call \emph{relative progress} and \emph{relative preservation}.

\begin{thm}[Relative Preservation]
  \label{thm:relative-pres}
  If the local type system enjoys \textsc{Preservation}, then for every choreography~$C_1$ such that $\ChorTyping{\Gamma}{\Delta}{C_1}{\ChorTypeTau}$ and $\ChorStepsTo{R}{B}{C_1}{C_2}$,
  $\ChorTyping{\Gamma}{\Delta}{C_2}{\ChorTypeTau}$.
\end{thm}

\begin{thm}[Relative Progress]
  \label{thm:relative-prog}
  If the local type system enjoys \textsc{Boolean Invertability} and \textsc{Progress}, then for every choreography $C_1$ such that $\ChorTyping{\cdot}{\cdot}{C_1}{\ChorTypeTau}$, either $C_1$ is a choreography value or there are some $R$, $B$, and $C_2$ such that $\ChorStepsTo{R}{B}{C_1}{C_2}$.
\end{thm}

\begin{cor}[Relative Soundness]
  \label{thm:relative-sound}
  If the local type system is sound, as defined in Section~\ref{sec:typed-local-program}, then the choreographic type system enjoys progress and preservation.
\end{cor}

By dividing up the result into more-precise theorems, we are able to get some guarantees even when the local type system is not sound.
For instance, the unityped $\lambda$-calculus example (Example~\ref{ex:unityped-lambda}) is not sound, but does guarantee preservation.
Thus, Theorem~\ref{thm:relative-pres} allows us to lift preservation to the choreographic system.
In this case, we do not get choreographic progress, intuitively because we may get stuck when trying to evaluate a local expression, or when an if expression tries to discriminate on a non-boolean value.


%% file: eqns.tex
\section{Equational Reasoning}
\label{sec:equational-reasoning}

Choreographies represent collections of programs running in parallel.
In order to represent these programs serially, we are forced to decide what behavior to write first.
For instance, consider the program
$
\Send{\GenericLocI}{2+3}{\GenericLocII}{x}{
  \Send{\GenericLocIII}{3*4}{\GenericLocIV}{y}{C}}
$.
The program represents \GenericLocI sending a message to \GenericLocII while \GenericLocIII sends a message to \GenericLocIV.
We could have just as well represented that program as 
$
\Send{\GenericLocIII}{3*4}{\GenericLocIV}{y}{
  \Send{\GenericLocI}{2+3}{\GenericLocII}{x}{C}}
$.
Since these choreographies represent the same collection of programs, all of our constructs treat them the same way.
In prior work on choreographies, this fact is typically formalized using a separate notion of \emph{equivalence}, which says when two choreographies represent the same collection of programs~\cite{Cruz-FilipeM17,Cruz-FilipeM17c,Montesi13,LaneseMZ13}.
Following that tradition, we define a notion of structural equivalence for Pirouette and study its properties.


\begin{figure}
  \begin{mathpar}
    \infer[SwapSendSend]{\GenericLocI \neq \GenericLocIII\\ \GenericLocII \neq \GenericLocIII \\ \GenericLocI \neq \GenericLocIV\\ \GenericLocII \neq \GenericLocIV}{
      \ChorEquiv{\Send*{\GenericLocI}{e_1}{\GenericLocII}{x}{\Send{\GenericLocIII}{e_2}{\GenericLocIV}{y}{C}}}
      {\Send*{\GenericLocIII}{e_2}{\GenericLocIV}{y}{\Send{\GenericLocI}{e_1}{\GenericLocII}{x}{C}}}
    } \and
    \infer[SwapSendSync]{\GenericLocI \neq \GenericLocIII\\ \GenericLocII \neq \GenericLocIII \\ \GenericLocI \neq \GenericLocIV\\ \GenericLocII \neq \GenericLocIV}{
      \ChorEquiv{\Send*{\GenericLocI}{e}{\GenericLocII}{x}{\Sync{\GenericLocIII}{\GenChoice}{\GenericLocIV}{C}}}
      {\Sync*{\GenericLocIII}{\GenChoice}{\GenericLocIV}{\Send{\GenericLocI}{e}{\GenericLocII}{x}{C}}}
    }\and
    \infer[SwapSendIf]{\GenericLocI \neq \GenericLocIII\\ \GenericLocII \neq \GenericLocIII}{
      \ChorEquiv{\Send*{\GenericLocI}{e_1}{\GenericLocII}{x}{\ChorIf*{\GenericLocIII}{e_2}{C_1}{C_2}}}
      {\ChorIf*{\GenericLocIII}{e_2}{\Send{\GenericLocI}{e_1}{\GenericLocII}{x}{C_1}}{\Send{\GenericLocI}{e_1}{\GenericLocII}{x}{C_2}}}
    }\and
    \infer[SwapSyncSync]{\GenericLocI \neq \GenericLocIII\\ \GenericLocII \neq \GenericLocIII \\ \GenericLocI \neq \GenericLocIV\\ \GenericLocII \neq \GenericLocIV}{
      \ChorEquiv{\Sync*{\GenericLocI}{\GenChoice}{\GenericLocII}{\Sync{\GenericLocIII}{\GenChoiceII}{\GenericLocIV}{C}}}
      {\Sync*{\GenericLocIII}{\GenChoiceII}{\GenericLocIV}{\Sync{\GenericLocI}{\GenChoice}{\GenericLocII}{C}}}
    }\and
    \infer[SwapSyncIf]{\GenericLocI \neq \GenericLocIII\\ \GenericLocII \neq \GenericLocIII}{
      \ChorEquiv{\Sync*{\GenericLocI}{\GenChoice}{\GenericLocII}{\ChorIf*{\GenericLocIII}{e}{C_1}{C_2}}}
      {\ChorIf*{\GenericLocIII}{e}{\Sync{\GenericLocI}{\GenChoice}{\GenericLocII}{C_1}}{\Sync{\GenericLocI}{\GenChoice}{\GenericLocII}{C_2}}}
    }\and
    \infer[SwapIfIf]{\GenericLocI \neq \GenericLocII}{
      \ChorEquiv{\ChorIf*{\GenericLocI}{e_1}{\ChorIf*{\GenericLocII}{e_2}{C_1}{C_2}}{\ChorIf*{\GenericLocII}{e_2}{C_3}{C_4}}}
      {\ChorIf*{\GenericLocII}{e_2}{\ChorIf*{\GenericLocI}{e_1}{C_1}{C_3}}{\ChorIf*{\GenericLocI}{e_1}{C_2}{C_4}}}
    }  
  \end{mathpar}
  \caption{Selected Choreography Equivalence Rules}
  \label{fig:chor-equiv}
\end{figure}

We define choreography equivalence as the smallest equivalence relation which is also a congruence and satisfies the rules in Figure~\ref{fig:chor-equiv}.
The complete formal definition can be found in \iftr{}Appendix~\ref{sec:full-chor-equiv}\else{}the accompanying technical report~\cite{HirschG21}\fi{}.

Choreography equivalence respects types, even for open programs:
\begin{thm}[Equivalence Respects Types]
  \label{thm:equiv-types}
  If $\ChorTyping{\Gamma}{\Delta}{C_1}{\ChorTypeTau}$ and $\ChorEquiv{C_1}{C_2}$, then $\ChorTyping{\Gamma}{\Delta}{C_2}{\ChorTypeTau}$
\end{thm}

Interestingly, Theorem~\ref{thm:equiv-types} has an outsized influence on our system model.
In particular, in order to prove that the rule \textsc{SwapSendIf} respects types, we need to know that expressions have unique types.
After all, if $e_2$ is given type~$\ExprTypeTauI$ in the \True branch, but type~$\ExprTypeTauII$ in the \False branch, there might not be a type that we could assign to $e_2$ that makes both branches type check.
Including a typing annotation on send statements might solve this problem.
However, this would require mixing the type system with the syntax of choreographies.
While this is a reasonable choice in many situations, here we choose to keep them separate, since that makes it easy to tell when our results rely on the type system and when they do not.

Finally, our operational semantics respects equivalence, allowing us to prove the following simulation theorem:
\begin{thm}[Operational Semantics Simulates Equivalence]
  If $\ChorStepsTo{R}{B}{C_1}{C_2}$ and $\ChorEquiv{C_1}{C_1'}$, then there is a $C_2'$ such that $\ChorEquiv{C_2}{C_2'}$ and $\ChorStepsTo{R}{B}{C_1'}{C_2'}$.  
\end{thm}

Equivalence can be used to define a new, seemingly simpler operational semantics for choreographies. 
This new semantics, which we write $\EquivStep$, has no redices or block sets and it replaces internal steps from our semantics with the following rule:
$$
\infer[EquivStep]{\ChorEquiv{C_1}{C_1'}\\ \EquivStepsTo{C_1'}{C_2'}\\ \ChorEquiv{C_2'}{C_2}}{\EquivStepsTo{C_1}{C_2}}
$$
We formalize this semantics in \iftr{}Appendix~\ref{sec:equiv-based-semant}\else{}the accompanying technical report~\cite{HirschG21}\fi{}.

In fact, much of the prior work on choreographies defined their operational semantics in precisely this way~\citep{Cruz-FilipeM17,Cruz-FilipeM17c,CarboneMS14,Montesi13,LaneseMZ13}.
While this new semantics is good for prior work, it is too weak for Pirouette, which allows reduction of local expressions.
To see why, consider the program
$
\Send{\GenericLocI}{e_1}{\GenericLocII}{x}{
  \Own{\GenericLocIII}{e_2}
}
$.
It should be possible for \GenericLocIII to evaluate its return value, even though \GenericLocI has not yet sent its message to \GenericLocII.
However, under the equivalence-based semantics, there is no way to reduce $e_2$, since we cannot use an equivalence to bring it up to the top.

A similar problem appears for sends.
Consider the program
$
\Send{\GenericLocI}{e_1}{\GenericLocII}{x}{
  \Send{\GenericLocIII}{e_2}{\GenericLocII}{x}{
    C
  }
}
$.
In this example, both \GenericLocI and \GenericLocIII are trying to send a message to \GenericLocII.
If $e_2$ can be reduced further, then \GenericLocIII ought to be able to reduce $e_2$ while waiting for \GenericLocII to be ready to receive the second message.
However, in the equivalence-based semantics, we cannot use \textsc{SwapSendSend} to bring up the second send, since that would swap the order of \GenericLocII's receives.

Finally, consider programs that contain function applications, such as
$$
\Send*{\GenericLocI}{e_1}{\GenericLocII}{x}{
  \AppGlobal{(\FunGlobal{F}{X}{X})}{(\Send{\GenericLocIII}{e_2}{\GenericLocIV}{y}{\Own{\GenericLocIV}{y}})}
}
$$
This program allows \GenericLocI to send a message to \GenericLocII and then applies the identity function to another choreography.
Importantly, the argument choreography does not mention \GenericLocI or \GenericLocII.
Therefore, \GenericLocIII should be able to reduce $e_2$ before \GenericLocI completes its send.
However, it is not possible to do this in the equivalence-based definition mentioned before.

This makes reasoning with $\Equiv$ much less powerful here than in previous choreographic systems.
However, the difficulty is limited to the problems mentioned above.
\begin{thm}[Weak Semantics]
  Let $\WeakStep{R}{B}$ be the relation obtained by modifying the relation $\ChorStep{R}{B}$ as follows:
  \begin{itemize}
  \item $\Own{\GenericLoc}{e}$ can only reduce when the block set is empty,
  \item messages can only be reduced when neither the sender nor the receiver are in the block set, and
  \item in $\DefLocal{\GenericLoc}{x}{C_1}{C_2}$ and function application, subchoreographies can only reduce when the block set is empty.
  \end{itemize}
  (The semantics $\WeakStep{R}{B}$ is formalized in \iftr{}Appendix~\ref{sec:weak-block-set}\else{}the accompanying technical report~\cite{HirschG21}\fi{}.)
  Then whenever $\WeakStepsTo{R}{B}{C_1}{C_2}$, $\EquivStepsTo{C_1}{C_2}$.
  Moreover, whenever $\EquivStepsTo{C_1}{C_2}$, there is a redex~$R$ and choreography $C_2'$ such that $\WeakStepsTo{R}{\emptyset}{C_1}{C_2'}$ and $\ChorEquiv{C_2}{C_2'}$.
\end{thm}

While this makes our connection to previous work more clear, we prefer to work with the $\ChorStep{R}{B}$ relation defined in Section~\ref{sec:funct-chor} due to its extra power.


%% file: compiling.tex
\section{Endpoint Projection}
\label{sec:comp-chor}

While Pirouette programs are designed to represent a collection of concurrently-executing programs, so far that has been a guiding intuition rather than a formal property.
In order to change that, we define the \emph{endpoint projection} operation, which \emph{extracts a program} for each location from a choreography, if it is possible to do so.
This extracted program is expressed in a language called the \emph{control language} which features local execution and explicit constructs for message passing.
The extracted programs of all locations are composed in parallel.
Here, we explain the control language, then explain the extraction and finally show that the parallel composition of all extracted programs reflects and preserves the operational semantics of the choreography.

\subsection{The Control Language}
\label{sec:control-language}

\begin{figure}
  \begin{syntax}
    \category[Control Expression]{E}
    \alternative{X}
    \alternative{\ContFunLocal{F}{x}{E}}
    \alternative{\ContFunGlobal{F}{X}{E}}
    \alternative{\ContAppLocal{E}{e}}
    \alternative{\ContAppGlobal{E_1}{E_2}}
    \alternativeLine{\Unit}
    \alternative{\Ret(e)}
    \alternative{\LetRet{x}{E_1}{E_2}}
    \alternativeLine{\ContSend{e}{\GenericLoc}{E}}
    \alternative{\ContReceive{x}{\GenericLoc}{E}}
    \alternativeLine{\ContIf{e}{E_1}{E_2}}
    \alternative{\ContChoose{\GenChoice}{\GenericLoc}{E}}
    \alternativeLine{\AllowChoiceLR{\GenericLoc}{E_1}{E_2}}
    \category[Systems]{\Pi} \alternative{\LocBind{\GenericLocI}{E_1} \mathrel{\Par} \cdots \mathrel{\Par} \LocBind{\GenericLocN}{E_n}}
  \end{syntax}
  
  \caption{Control Language Syntax}
  \label{fig:control-syn}
\end{figure}

Our control language (Figure~\ref{fig:control-syn}) is a concurrent $\lambda$-calculus where messages are values of local programs.
It is inspired both by work on process calculi and by concurrent ML.

Like with Pirouette, control programs have two types of variable: local variables and control variables.
We write control variables with capital letters, because they play a role similar to that played by choreography variables.
Local variables are the variables of local programs.
There are correspondingly two types of functions, local functions and global functions (written $\ContFunLocal{F}{x}{E}$ and $\ContFunGlobal{F}{X}{E}$, respectively).
Note that because in the control language---unlike in Pirouette---every local program is at the same location, local substitution does not take location into account.

Control programs can return the result of evaluating a local program, which we write $\Ret(e)$.
We can use the result of such a program in another program using the syntax $\LetRet{x}{E_1}{E_2}$.
However, unlike choreographies, control programs can also return the trivial value $\Unit$.
This is used for control programs that do not have a return value on them.

Communication happens between control programs composed in parallel through explicit send and receive commands.
We will see later how parallel composition works, and how communication takes place.

There are two forms of branching in the control language.
``If'' statements are standard, and are a sequential form of branching.
We also have \emph{external choice}, which is a distributed form of branching.
The program $\AllowChoiceLR{\GenericLoc}{E_1}{E_2}$ represents allowing $\GenericLoc$ to choose which branch to take: $E_1$, labeled \LChoice, or $E_2$, labeled \RChoice.
The syntax $\ContChoose{\GenChoice}{\GenericLoc}{E}$ represents telling $\GenericLoc$ to take the branch labeled $\GenChoice$.

We refer to the parallel composition of a control program for each location as a \emph{system} (we use the symbol $\Pi$ to refer to systems).
The notation $\LocBind{\GenericLoc}{E}$ says that program $E$ is running on the node $\GenericLoc$.
As can be seen in Figure~\ref{fig:control-syn}, a system is a finite parallel compositions of such $\LocBind{\GenericLoc}{E}$.

We often use the syntax $\BigPar{\ell \in \mathfrak{L}}{E_\ell}$ to refer to a system where $\mathfrak{L}$ is a finite set of locations and $E_{-}$ is a function from locations to control program expressions.

\paragraph{Location semantics}
In defining the operational semantics of systems, two syntactic operations will be useful.
The first is \emph{system lookup} (written $\Pi(\LocA)$), which refers to the control program bound to a particular location~$\LocA$.
The second is \emph{system update} (written $\Subst{\Pi}{\SingleSubst{\LocA}{E}}$), which replaces the program bound to~$\LocA$.
We define them as follows:
\begin{mathpar}
  \left(\BigPar{\GenericLoc \in \mathfrak{L}}{E_{\GenericLoc}}\right)(\LocA) = E_{\LocA} \and
  \Subst{\left(\BigPar{\GenericLoc \in \mathfrak{L}}{E_{\GenericLoc}}\right)}{\SingleSubst{\LocA}{E}} = \BigPar{\GenericLoc \in \mathfrak{L}}{E_{\GenericLoc}'}~\text{where}~E_{\GenericLoc}' = \left\{\begin{array}{ll}E & \GenericLoc = \LocA\\E_{\GenericLoc} & \text{otherwise}\end{array}\right.
\end{mathpar}

\begin{figure}
  \begin{syntax}
    \category[Label]{l} \alternative{\TauLabel}
    \alternative{\SendLabel{v}{\GenericLoc}}
    \alternative{\RecvLabel{\GenericLoc}{v}}
    \alternative{\ChooseLabel{\GenChoice}{\GenericLoc}}
    \alternative{\AllowChoiceLabel{\GenericLoc}{\GenChoice}}
    \alternative{\SyncTauLabel}
    \alternative{\FunLabel(l)}
    \alternative{\ArgLabel(l)}
  \end{syntax}
  \begin{mathpar}
    \infer[SendE]{
      \ExprStepsTo{e_1}{e_2}
    }{
      \ContStepsTo{\TauLabel}{\ContSend{e_1}{\GenericLoc}{E}}{\ContSend{e_2}{\GenericLoc}{E}}
    }\\
    \infer[SendV]{
      \ExprValue{v}
    }{
      \ContStepsTo{\SendLabel{v}{\GenericLoc}}{\ContSend{v}{\GenericLoc}{E}}{E}
    } \and
    \infer[RecvV]{
      \ExprValue{v}
    }{
      \ContStepsTo{\RecvLabel{\GenericLoc}{v}}{\ContReceive{x}{\GenericLoc}{E}}{\Subst{E}{\SingleSubst{x}{v}}}
    } \and
    \infer[Choose]{ }{
      \ContStepsTo{\ChooseLabel{\GenChoice}{\GenericLoc}}{\ContChoose{\GenChoice}{\GenericLoc}{E}}{E}
    } \and
    \infer[AllowChoiceL]{ }{
      \ContStepsTo{\AllowChoiceLabel{\GenericLoc}{\LChoice}}
      {\AllowChoiceLR{\GenericLoc}{E_1}{E_2}}{E_1}
    } \and
    \infer[LetRet]{
      \ExprValue{v}
    }{
      \ContStepsTo{\SyncTauLabel}{\LetRet{x}{\Ret(v)}{E_2}}{\Subst{E_2}{\SingleSubst{x}{v}}}
    } \and
    \infer[AppLocal]{
      \ExprValue{v}
    }{
      \ContStepsTo{\SyncTauLabel}{\ContAppLocal{\left(\ContFunLocal{F}{x}{E}\right)}{v}}{\Subst{E}{\SingleSubst{x}{v}}}
    } \and
    \infer[AppGlobal]{
      \ContValue(V)
    }{
      \ContStepsTo{\SyncTauLabel}{\ContAppGlobal{\left(\ContFunGlobal{F}{X}{E}\right)}{V}}{\Subst{E}{\SingleSubst{X}{V}}}
    }
  \end{mathpar}
  \caption[Control Programs Semantics]{Control Programs Semantics (Selected Rules)}
  \label{fig:control-sem}
\end{figure}

The semantics of control programs is given via a labeled transition system.
This allows systems to match up corresponding rules in their semantics.
The syntax of labels and selected rules can be found in Figure~\ref{fig:control-sem}, and the full set of rules can be found in \iftr{}Appendix~\ref{sec:contr-lang-oper}\else{}the accompanying technical report~\cite{HirschG21}\fi.

Internal steps that do not interact with the outside are given the label~$\TauLabel$
\footnote{It is standard to use $\tau$ to refer to internal steps.
  However, $\tau$~already represents Pirouette types.
  We thus use $\iota$ for ``internal.''}.
For instance, the rule \textsc{SendE} takes a local step in a message to be sent to $\GenericLoc$, which is an internal step.
Every step of a local program corresponds to an $\TauLabel$~step.

Sends and receives are labeled with matching labels: $\SendLabel{v}{\GenericLocII}$ for sends, and $\RecvLabel{\GenericLocI}{v}$ for receives.
Note that from the perspective of a single control-language program, receives are treated nondeterministically---any value could be received.
Our system semantics will force sends and receives to match up.
Similarly, external choice and its resolution have matching labels: $\AllowChoiceLabel{\GenericLoc}{\GenChoice}$ for external choice and $\ChooseLabel{\GenChoice}{\GenericLoc}$ for  its resolution.

In choreographies, all participants must $\beta$-reduce function applications together.
For instance, consider reducing a local function:
$$\ChorStepsTo{\RAppLocal{\GenericLoc}{v}}{\emptyset}{\AppLocal{\GenericLoc}{(\FunLocal{\GenericLoc}{F}{x}{C})}{e}}{\Subst{\Subst{C}{\SingleChorExprSubst{\GenericLoc}{x}{v}}}{\SingleSubst{F}{\FunLocal{\GenericLoc}{F}{x}{C}}}}$$
Because the choreography steps from a choreography with a function application to one without, \emph{every} location's control program changes.
In order to accommodate this, we use a new label,~$\SyncTauLabel$.
The only three rules labeled with~$\SyncTauLabel$ are \textsc{LetRet}, \textsc{AppLocal}, and \textsc{AppGlobal}, all of which can be found in Figure~\ref{fig:control-sem}.

\begin{figure}
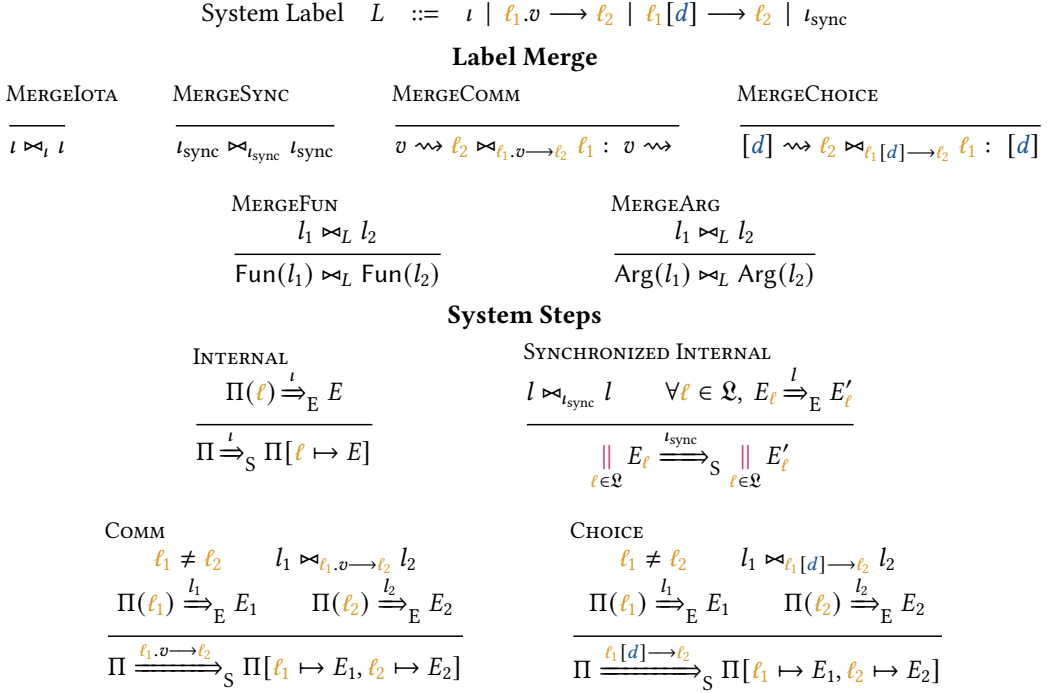

  \begin{syntax}
    \category[System Label]{L}
    \alternative{\SysTau}
    \alternative{\CommLabel{\GenericLocI}{v}{\GenericLocII}}
    \alternative{\ChoiceLabel{\GenericLocI}{\GenChoice}{\GenericLocII}}
    \alternative{\SysSyncTau}
  \end{syntax}
  \textbf{Label Merge}
  \begin{mathpar}
    \infer[MergeIota]{ }{\LabelMerge{\TauLabel}{\TauLabel}{\SysTau}} \and
    \infer[MergeSync]{ }{\LabelMerge{\SyncTauLabel}{\SyncTauLabel}{\SysSyncTau}} \and
    \infer[MergeComm]{ }{\LabelMerge{\SendLabel{v}{\GenericLocII}}{\RecvLabel{\GenericLocI}{v}}{\CommLabel{\GenericLocI}{v}{\GenericLocII}}} \and
    \infer[MergeChoice]{ }{\LabelMerge{\ChooseLabel{\GenChoice}{\GenericLocII}}{\AllowChoiceLabel{\GenericLocI}{\GenChoice}}{\ChoiceLabel{\GenericLocI}{\GenChoice}{\GenericLocII}}} \and
    \infer[MergeFun]{
      \LabelMerge{l_1}{l_2}{L}
    }{
      \LabelMerge{\FunLabel(l_1)}{\FunLabel(l_2)}{L}
    } \and
    \infer[MergeArg]{
      \LabelMerge{l_1}{l_2}{L}
    }{
      \LabelMerge{\ArgLabel(l_1)}{\ArgLabel(l_2)}{L}
    }
  \end{mathpar}
  \textbf{System Steps}
  \begin{mathpar}
    \infer[Internal]{
      \ContStepsTo{\TauLabel}{\Pi(\GenericLoc)}{E}
    }{
      \SysStepsTo{\SysTau}{\Pi}{\Subst{\Pi}{\SingleSubst{\GenericLoc}{E}}}
    } \and
    \infer[Synchronized Internal]{
      \LabelMerge{l}{l}{\SyncTauLabel}\\
      \forall \GenericLoc \in \mathfrak{L},~\ContStepsTo{l}{E_{\GenericLoc}}{E'_{\GenericLoc}}
    }{
      \SysStepsTo{\SysSyncTau}{\BigPar{\GenericLoc \in \mathfrak{L}}{E_{\GenericLoc}}}{\BigPar{\GenericLoc \in \mathfrak{L}}{E_{\GenericLoc}'}}
    } \\
    \infer[Comm]{
      \GenericLocI \neq \GenericLocII\\
      \LabelMerge{l_1}{l_2}{\CommLabel{\GenericLocI}{v}{\GenericLocII}}\\\\
      \ContStepsTo{l_1}{\Pi(\GenericLocI)}{E_1}\\
      \ContStepsTo{l_2}{\Pi(\GenericLocII)}{E_2}
    }{
      \SysStepsTo{\CommLabel{\GenericLocI}{v}{\GenericLocII}}{\Pi}{\Subst{\Pi}{\SingleSubst{\GenericLocI}{E_1}, \SingleSubst{\GenericLocII}{E_2}}}
    } \and
    \infer[Choice]{
      \GenericLocI \neq \GenericLocII\\
      \LabelMerge{l_1}{l_2}{\ChoiceLabel{\GenericLocI}{\GenChoice}{\GenericLocII}}\\\\
      \ContStepsTo{l_1}{\Pi(\GenericLocI)}{E_1}\\
      \ContStepsTo{l_2}{\Pi(\GenericLocII)}{E_2}
    }{
      \SysStepsTo{\ChoiceLabel{\GenericLocI}{\GenChoice}{\GenericLocII}}{\Pi}{\Subst{\Pi}{\SingleSubst{\GenericLocI}{E_1}, \SingleSubst{\GenericLocII}{E_2}}}
    }
  \end{mathpar}
  \caption{System Semantics}
  \label{fig:sys-sem}
\end{figure}

\paragraph{System semantics}
Systems are also given semantics via a labeled transition system.
The system labels arise from a merging operator on control-language labels, where $\LabelMerge{l_1}{l_2}{L}$ ensures that labels $l_1$ and $l_2$ match, producing an output system label~$L$.
It also ensures that, in a function application, either both steps reduce the function or both reduce its argument.
The syntax of system labels, label merge operator, and system semantics can all be found in Figure~\ref{fig:sys-sem}.

The labels~$\TauLabel$ and~$\SyncTauLabel$ both refer to internal steps of a production. 
Hence, they can be merged with themselves to yield the corresponding system label.
The rule \textsc{Internal} allows any location~\GenericLoc to take an internal step without interfering with any other location.
Synchronized steps require the use of the \textsc{Synchronized Internal} rule, which requires that \emph{every} location take a \SyncTauLabel~step.

Send and receive labels must be matched together.
If $\GenericLocI$ sends $v$ to $\GenericLocII$, their labels merge together to a system label $\CommLabel{\GenericLocI}{v}{\GenericLocII}$ in rule \textsc{MergeComm}.
The rule \textsc{Comm} then allows both $\GenericLocI$ and $\GenericLocII$ to take their corresponding steps together, without interfering with any other locations.
The \textsc{Choice} rule behaves similarly, but for choice-based branching.

\subsection{Merging Control Programs}
\label{sec:merg-contr-progr}

Our goal is to extract a control program for every location in a choreography \emph{compositionally}.
However, compositionality  is greatly complicated by if branches.
To see why, consider the following choreography:
$$
\ChorIf*{\GenericLocI}{e}{
  \Send{\GenericLocII}{3}{\GenericLocI}{x}{
    \Sync{\GenericLocI}{\LChoice}{\GenericLocII}{
      \Own{\GenericLocII}{0}
    }
  }
}{
  \Send{\GenericLocII}{3}{\GenericLocI}{x}{
    \Sync{\GenericLocI}{\RChoice}{\GenericLocII}{
      \Own{\GenericLocII}{1}
    }
  }
}    
$$

Intuitively, we want $\GenericLocII$ to be associated with the control program
$$
\ContSend{3}{\GenericLocI}{
  \AllowChoiceLR{\GenericLocI}{\Ret(0)}{\Ret(1)}
}
$$
However, when defining our procedure formally, we want to extract a program for each branch of the if expression, and combine them together to get the final program.

This leads to two issues.
First, we need to be able to define a program for each branch.
Second, we need to be able to merge those two programs into a single program.

To see why it is difficult to define a program for each branch, consider the true branch of the program above.
We know that $\GenericLocI$ will send a synchronization message to $\GenericLocII$.
Thus, $\GenericLocII$ must allow $\GenericLocI$ to make a choice for it, as we saw earlier.
However, here we only have the $\LChoice$ branch available; the $\RChoice$ branch will not be available until we merge.

To solve this problem, we add one-branch choice constructs to our control language:
$$\AllowChoiceL*{\GenericLoc}{E}\hspace{3em}\text{and}\hspace{3.1em}\AllowChoiceR*{\GenericLoc}{E}$$
These act precisely like the two-choice construct, except that they only allow their one branch to be taken.
With these, we now have a program for $\GenericLocII$ we intend to extract from the branch above:
$$
\ContSend{3}{\GenericLocI}{
  \AllowChoiceL{\GenericLocI}{\Ret(0)}
}
$$

\begin{figure}
  \begin{mathpar}
  \ContMerge{\left(\AllowChoiceL*{\GenericLoc}{E_1}\right)}{\left(\AllowChoiceL*{\GenericLoc}{E_2}\right)} \triangleq \AllowChoiceL*{\GenericLoc}{\ContMerge{E_1}{E_2}} \and
  \ContMerge{\left(\AllowChoiceL*{\GenericLoc}{E_1}\right)}{\left(\AllowChoiceR*{\GenericLoc}{E_2}\right)} \triangleq \AllowChoiceLR*{\GenericLoc}{E_1}{E_2} \and
  \ContMerge{\left(\AllowChoiceL*{\GenericLoc}{E_1}\right)}{\left(\AllowChoiceLR*{\GenericLoc}{E_{2,1}}{E_{2,2}}\right)} \triangleq \AllowChoiceLR*{\GenericLoc}{\ContMerge{E_1}{E_{2,1}}}{E_{2,2}}
  \end{mathpar}
  \caption{Merge Operator Definition (Selected Parts)}
  \label{fig:cont-merge}
\end{figure}

We now need a way to merge the extracted programs from each branch into a single program.
We define a \emph{partial} merge operator~\ContMergeOperator, which ensures that two programs are the same until they allow both choices.
We adopt notation from computability theory and write $\Undefed{f(x)}$ when $f$ is a partial function to denote that the function is undefined.
So we would write $\Undefed{\ContMerge{E_1}{E_2}}$ if the merge of $E_1$ and $E_2$ is undefined.
We further write $\Defed{f(x)}$ (so $\Defed{\ContMerge{E_1}{E_2}}$) if $f(x)$ is defined, but we do not care about the value.
In Figure~\ref{fig:cont-merge}, you can find the definition of the merge operator when the left-hand side is $\AllowChoiceL{\GenericLoc}{E_1}$.
The remaining parts of the definition can be found in \iftr{}Appendix~\ref{sec:contr-progr-merge}\else{}the accompanying technical report~\cite{HirschG21}\fi{}.
\iftr{}Additional properties can be found in Appendix~\ref{sec:addit-prop-merg}.\fi

\subsection{Endpoint Projection, Defined}
\label{sec:endp-proj-defin}

We are now ready to define endpoint projection, or~EPP.
Merging is an important part of the definition of EPP, and EPP inherits partiality from merging.
We continue to use $\Undefed{f(x)}$ if $f(x)$ is undefined and $\Defed{f(x)}$ if $f(x)$ is defined, but we do not care about the value.

Endpoint projection is defined as follows: 
\begin{mathparpagebreakable}
  \ProjectC{\Own{\GenericLocI}{e}}{\GenericLocII} = \left\{\begin{array}{ll} \Ret(e) & \text{if}~\GenericLocI = \GenericLocII\\\Unit & \text{otherwise}\end{array}\right.\and
  \ProjectC{X}{\GenericLoc} = X \and
  \ProjectC{\Send{\GenericLocI}{e}{\GenericLocII}{x}{C}}{\GenericLocIII} =
    \left\{
      \begin{array}{ll}
        \uparrow & \text{if}~\GenericLocI=\GenericLocII=\GenericLocIII\\
        \ContSend{e}{\GenericLocII}{\ProjectC{C}{\GenericLocIII}} & \text{if}~\GenericLocI=\GenericLocIII \neq \GenericLocII\\
        \ContReceive{x}{\GenericLocI}{\ProjectC{C}{\GenericLocIII}} & \text{if}~\GenericLocI\neq\GenericLocIII=\GenericLocII\\
        \ProjectC{C}{\GenericLocIII} & \text{if}~\GenericLocI\neq\GenericLocIII \text{ and } \GenericLocII\neq\GenericLocIII
      \end{array}
    \right. \and
  \ProjectC{\ChorIf{\GenericLocI}{e}{C_1}{C_2}}{\GenericLocII} =
    \left\{
      \begin{array}{ll}
        \ContIf{e}{\ProjectC{C_1}{\GenericLocII}}{\ProjectC{C_2}{\GenericLocII}} & \text{if}~\GenericLocI=\GenericLocII\\
        \ContMerge{\ProjectC{C_1}{\GenericLocII}}{\ProjectC{C_2}{\GenericLocII}} & \text{otherwise}
      \end{array}
    \right. \and
  \ProjectC{\Sync{\GenericLocI}{\GenChoice}{\GenericLocII}{C}}{\GenericLocIII} =
    \left\{
      \begin{array}{ll}
        \uparrow & \text{if}~\GenericLocI=\GenericLocII=\GenericLocIII\\
        \ContChoose{\GenChoice}{\GenericLocII}{\ProjectC{C}{\GenericLocIII}} & \text{if}~\GenericLocI=\GenericLocIII \neq \GenericLocII\\
        \AllowChoiceL{\GenericLocI}{\ProjectC{C}{\GenericLocIII}} & \text{if}~\GenericLocI\neq\GenericLocIII=\GenericLocII~\text{and}~\GenChoice=\LChoice\\
        \AllowChoiceR{\GenericLocI}{\ProjectC{C}{\GenericLocIII}} & \text{if}~\GenericLocI\neq\GenericLocIII=\GenericLocII~\text{and}~\GenChoice=\RChoice\\
        \ProjectC{C}{\GenericLocIII} & \text{otherwise}
      \end{array}
    \right. \and
  \ProjectC{\DefLocal{\GenericLocI}{x}{C_1}{C_2}}{\GenericLocII} =
    \left\{
      \begin{array}{ll}
        \LetRet{x}{\ProjectC{C_1}{\GenericLocII}}{\ProjectC{C_2}{\GenericLocII}} & \text{if}~\GenericLocI=\GenericLocII\\
        \AppGlobal{(\ContFunGlobal{F}{X}{\ProjectC{C_2}{\GenericLocII}})}{\ProjectC{C_1}{\GenericLocII}} & \text{where $F,X$ are fresh, otherwise}
      \end{array}
    \right. \and
  \ProjectC{\FunLocal{\GenericLocI}{F}{x}{C}}{\GenericLocII} =
    \left\{
      \begin{array}{ll}
        \ContFunLocal{F}{x}{\ProjectC{C}{\GenericLocII}} & \text{if}~\GenericLocI = \GenericLocII\\
        \ContFunGlobal{F}{X}{\ProjectC{C}{\GenericLocII}} & \text{where $X$ is fresh, otherwise}
      \end{array}
    \right. \and
  \ProjectC{\AppLocal{\GenericLocI}{C}{e}}{\GenericLocII} =
    \left\{
      \begin{array}{ll}
        \ContAppLocal{\ProjectC{C}{\GenericLocII}}{e} & \text{if}~\GenericLocI=\GenericLocII\\
        \ContAppGlobal{\ProjectC{C}{\GenericLocII}}{\Unit} & \text{otherwise}
      \end{array}
    \right. \and
  \ProjectC{\FunGlobal{F}{X}{C}}{\GenericLoc} = \ContFunGlobal{F}{X}{\ProjectC{C}{\GenericLoc}} \and
  \ProjectC{\AppGlobal{C_1}{C_2}}{\GenericLoc} = \ContAppGlobal{\ProjectC{C_1}{\GenericLoc}}{\ProjectC{C_2}{\GenericLoc}}  
\end{mathparpagebreakable}

Note a simple design principle: local expressions owned by a location other than the one being projected to are projected to $\Unit$.
This allows us to define a control program with the same control flow, but which does not know about the precise local expressions held at other locations.

This definition tells us what the control~program for a single location is.
However, we are interested in a system of control~programs.
We can lift the single-location definition to a multi-location system definition:
$
\ProjectC{C}{\mathfrak{L}} = \BigPar{\GenericLoc \in \mathfrak{L}}{\LocBind{\GenericLoc}{\ProjectC{C}{\GenericLoc}}}
$.

One way of thinking about $\ProjectC{C}{\GenericLoc}$ is that it gives $\GenericLoc$'s view of $C$.
From this perspective, it makes sense to ask what $\GenericLoc$'s view of a step of computation is.
We can provide this by projecting a choreography redex to a control-language label, as follows:
\begin{mathparpagebreakable}
  \ProjectR{\RDone{\GenericLocI}{e_1}{e_2}}{\GenericLocII} =
    \left\{
      \begin{array}{ll}
        \TauLabel & \text{if}~\GenericLocI=\GenericLocII\\
        \uparrow & \text{otherwise}
      \end{array}
    \right. \and
  \ProjectR{\RIfE{\GenericLocI}{e_1}{e_2}}{\GenericLocII} =
    \left\{
      \begin{array}{ll}
        \TauLabel & \text{if}~\GenericLocI=\GenericLocII\\
        \uparrow & \text{otherwise}
      \end{array}
    \right. \and
  \ProjectR{\RIfT{\GenericLocI}}{\GenericLocII} =
    \left\{
      \begin{array}{ll}
        \TauLabel & \text{if}~\GenericLocI=\GenericLocII\\
        \uparrow & \text{otherwise}
      \end{array}
    \right. \and
  \ProjectR{\RIfF{\GenericLocI}}{\GenericLocII} =
    \left\{
      \begin{array}{ll}
        \TauLabel & \text{if}~\GenericLocI=\GenericLocII\\
        \uparrow & \text{otherwise}
      \end{array}
    \right. \and
  \ProjectR{\RAppLocalE{\GenericLocI}{e_1}{e_2}}{\GenericLocII} =
    \left\{
      \begin{array}{ll}
        \TauLabel & \text{if}~\GenericLocI=\GenericLocII\\
        \uparrow & \text{otherwise}
      \end{array}
    \right. \\
  \ProjectR{\RSendE{\GenericLocI}{e_1}{e_2}{\GenericLocII}}{\GenericLocIII} =
    \left\{
      \begin{array}{ll}
        \tau & \text{if}~\GenericLocI = \GenericLocIII\\
        \uparrow & \text{otherwise}
      \end{array}
    \right. \and
  \ProjectR{\RSendV{\GenericLocI}{v}{\GenericLocII}}{\GenericLocIII} =
    \left\{
      \begin{array}{ll}
        \SendLabel{v}{\GenericLocII} & \text{if}~\GenericLocI = \GenericLocIII \neq \GenericLocII\\
        \RecvLabel{\GenericLocI}{v} & \text{if}~\GenericLocI \neq \GenericLocIII = \GenericLocII\\
        \uparrow & \text{otherwise}
      \end{array}
    \right. \and
  \ProjectR{\RSync{\GenericLocI}{\GenChoice}{\GenericLocII}}{\GenericLocIII} =
    \left\{
      \begin{array}{ll}
        \ChooseLabel{\GenChoice}{\GenericLocII} & \text{if}~\GenericLocI = \GenericLocIII \neq \GenericLocII\\
        \AllowChoiceLabel{\GenericLocI}{\GenChoice} & \text{if}~\GenericLocI \neq \GenericLocIII = \GenericLocII\\
        \uparrow & \text{otherwise}
      \end{array}
    \right. \and
  \ProjectR{\RDefLocal{\GenericLocI}{v}}{\GenericLocII} = \SyncTauLabel \and
  \ProjectR{\RAppLocal{\GenericLocI}{v}}{\GenericLocII} = \SyncTauLabel \and
  \ProjectR{\RAppGlobal}{\GenericLoc} = \SyncTauLabel \and
  \ProjectR{\RArg(R)}{\GenericLoc} = \ArgLabel(\ProjectR{R}{\GenericLoc}) \and
  \ProjectR{\RFun(R)}{\GenericLoc} = \FunLabel(\ProjectR{R}{\GenericLoc}) \and
\end{mathparpagebreakable}

Just as we can ask what some location's view of a step of computation is, we can ask what a system's view of a step of computation is.
Interestingly, this does not rely on which locations are in the system, and it is therefore a total operation.
\begin{mathparpagebreakable}
  \CompileR{\RDone{\GenericLocI}{e_1}{e_2}} = \SysTau \and
  \CompileR{\RIfE{\GenericLocI}{e_1}{e_2}} = \SysTau \and
  \CompileR{\RIfT{\GenericLocI}} = \SysTau \and
  \CompileR{\RIfF{\GenericLocI}} = \SysTau \and
  \CompileR{\RAppLocalE{\GenericLocI}{e_1}{e_2}} = \SysTau \and
  \CompileR{\RSendE{\GenericLocI}{e_1}{e_2}{\GenericLocII}} = \SysTau \and
  \CompileR{\RSendV{\GenericLocI}{v}{\GenericLocII}} = \CommLabel{\GenericLocI}{v}{\GenericLocII} \and
  \CompileR{\RSync{\GenericLocI}{\GenChoice}{\GenericLocII}} = \ChoiceLabel{\GenericLocI}{\GenChoice}{\GenericLocII} \and
  \CompileR{\RDefLocal{\GenericLocI}{v}} = \SysSyncTau \and
  \CompileR{\RAppLocal{\GenericLocI}{v}} = \SysSyncTau \and
  \CompileR{\RAppGlobal} = \SysSyncTau \and
  \CompileR{\RArg(R)} = \CompileR{R} \and
  \CompileR{\RFun(R)} = \CompileR{R} \and
\end{mathparpagebreakable}  

\subsection{Properties of Endpoint Projection}
\label{sec:prop-endp-proj}

EPP is one of the most-important operations on choreographies.
It is what gives them a ground-truth interpretation as a parallel composition of programs.
In fact, without EPP it would be almost impossible to state one of our most-important theorems: deadlock freedom by construction (Theorem~\ref{thm:deadlock-freedom} below).

The first property we examine is how EPP treats equivalence.
As with every other operation on choreographies, we would like it if EPP treated equivalent choreographies the same.
Note that we have no notion of equivalence on control-language programs, so we get a strong notion of ``treating the same'':
\begin{thm}[Equivalence Begets Equality]
  \label{thm:equiv-equality}
  If $\ChorEquiv{C_1}{C_2}$, then $\ProjectC{C_1}{\GenericLoc} = \ProjectC{C_2}{\GenericLoc}$ for every $\GenericLoc \in \mathcal{L}$.
\end{thm}

Next we would like to examine the relationship between the semantics of a choreography and the semantics of its projection.
However, there is still one remaining disconnect between the semantics of choreographies and that of systems that comes into play.
Yet again, it has to do with the semantics of external choice.
To see the issue, consider the following example:
\begin{mathpar}
C_1 \triangleq \ChorIf*{\GenericLocI}{\True}{\Sync{\GenericLocI}{\LChoice}{\GenericLocII}{\Own{\GenericLocII}{0}}}{\Sync{\GenericLocI}{\RChoice}{\GenericLocII}{\Own{\GenericLocII}{1}}} \and
C_2 \triangleq \Sync{\GenericLocI}{\LChoice}{\GenericLocII}{\Own{\GenericLocII}{0}} \and
\ProjectC{C_1}{\GenericLocII} = \AllowChoiceLR*{\GenericLocI}{\Ret(0)}{\Ret(1)} \and
\ProjectC{C_2}{\GenericLocII} = \AllowChoiceL*{\GenericLocI}{\Ret(0)} \and
\ChorStepsTo{\RIfT{\GenericLocI}}{\emptyset}{C_1}{C_2}
\end{mathpar}
As you can see, by taking a choreography step which corresponds to a completely internal step on \GenericLocI, we have lost information about a possible path on \GenericLocII.
This comes because the choice of the path is up to \GenericLocI, who ``makes up their mind'' in that internal step.

One way to view this is from \GenericLocII's point of view.
In the program $C_1$, \GenericLocII has a nondeterministic program: a message will come in to tell them which of two branches to take.
This is evident in the semantics of the control~language, since $\ProjectC{C_1}{\GenericLocII}$ can take either of two steps.
However, $C_1$ is deterministic from \GenericLocII's point of view, since only one message is possible.
Hence, this step has resolved some nondeterminism.

We formalize this notion of ``the same program, but with some nondeterminism resolved'' in a new relation called~\LNDRel.
This is \emph{nearly} defined as the smallest relation that commutes with all of the control-language constructs and also fulfills the following extra rules:
\begin{mathpar}
  \infer{\LND{E_1}{E_{2,1}}}
  {\LND{\AllowChoiceL*{\GenericLoc}{E_1}}{\AllowChoiceLR*{\GenericLoc}{E_{2,1}}{E_{2,2}}}} \and
  \infer{\LND{E_1}{E_{2,2}}}
  {\LND{\AllowChoiceR*{\GenericLoc}{E_1}}{\AllowChoiceLR*{\GenericLoc}{E_{2,1}}{E_{2,2}}}}
\end{mathpar}
However, there is a small complication: functions are only related to themselves.
The full definition can be found in \iftr{}Appendix~\ref{sec:less-nond-relat}\else{}the accompanying technical report~\cite{HirschG21}\fi{}.
The relation~\LNDRel is a partial order.

We extend the $\LNDRel$~relation to systems pointwise, so $\LND{\Pi_1}{\Pi_2}$ if for every $\GenericLoc$ such that $\Defed{\Pi_1(\GenericLoc)}$, $\Defed{\Pi_2(\GenericLoc)}$ and $\LND{\Pi_1(\GenericLoc)}{\Pi_2(\GenericLoc)}$.
The following theorem relates $\LNDRel$ to the semantics of systems.
It may look complicated, but is not.
Its two bullet points respectively say: (1) If a less nondeterministic system takes a step, then that step is available to any more-nondeterministic system, and
(2) If a more nondeterministic system takes a step but the less nondeterministic system is the result of projecting some choreography, then we can take advantage of the fact that choices are always paired in choreographies to mimic the step in the less nondeterministic system.
\begin{thm}[Lifting and Lowering System Steps Across~\LNDRel]
  \label{thm:lift-lower-global}
  If $\LND{\Pi_1}{\Pi_2}$, then the following are both true:
  \begin{itemize}
  \item If $\SysStepsTo{L}{\Pi_1}{\Pi_1'}$, then there is a $\Pi_2'$ such that $\LND{\Pi_1'}{\Pi_2'}$ and $\SysStepsTo{L}{\Pi_2}{\Pi_2'}$.
  \item If $\SysStepsTo{L}{\Pi_2}{\Pi_2'}$ and $\Pi_1 = \ProjectC{C}{\mathfrak{L}}$, then there is a $\Pi_1'$ such that $\LND{\Pi_1'}{\Pi_2'}$ and $\SysStepsTo{L}{\Pi_1}{\Pi_1'}$.
  \end{itemize}
\end{thm}

Finally, the less-nondeterminism relation allows us to connect the semantics of choreographies with the semantics of our control~language.
Steps in the choreographies correspond to actions by one or more locations.
If a step involves a location, then that location's control program also takes a step; the label of that step is the projection of the redex of the choreographic step.
However, if a step does not involve a location, then that location's control program does not take a step and the projection of the redex for that location is undefined.
While an uninvolved location does not take a step, it may find its nondeterminism reduced as other locations ``make up their minds'' about what branch it should take in the future.
\begin{thm}[Local Completeness]
  \label{thm:loc-complete}
  If $\ChorStepsTo{R}{B}{C_1}{C_2}$, then for any location $\GenericLoc$, either (a)~\mbox{$\ContStepsTo{\ProjectR{R}{\GenericLoc}}{\ProjectC{C_1}{\GenericLoc}}{\ProjectC{C_2}{\GenericLoc}}$}, or (b)~$\Undefed{\ProjectR{R}{\GenericLoc}}$ and $\LND{\ProjectC{C_2}{\GenericLoc}}{\ProjectC{C_1}{\GenericLoc}}$.
\end{thm}

\begin{thm}[Global Completeness]
  \label{thm:global-complete}
  If $\ChorStepsTo{R}{B}{C_1}{C_2}$ and every location named in $R$ is in $\mathfrak{L}$, then there is a $\Pi$ such that $\SysStepsTo{\CompileR{R}}{\ProjectC{C_1}{\mathfrak{L}}}{\Pi}$ and $\LND{\ProjectC{C_2}{\mathfrak{L}}}{\Pi}$.
\end{thm}
The requirement that every location named in $R$ is in $\mathfrak{L}$ is important: if $R$ represents \GenericLocI sending a message to \GenericLocII, but \GenericLocII is not in $\mathfrak{L}$, then there is no one in $\ProjectC{C}{\mathfrak{L}}$ to do the receiving.
This blocks \GenericLocI from sending its message (since message passing is synchronous) and thus blocks the system step from taking place.
However, if both \GenericLocI and \GenericLocII are in $\mathfrak{L}$, then Theorem~\ref{thm:loc-complete} tells us that \GenericLocII is guaranteed to be ready to receive \GenericLocI's message.

A similar difficulty comes when trying to go in the other direction: it is not enough to know that \GenericLocI's projection sends a message for a choreography to take a step; we must also know that \GenericLocII's projection receives it.
This means that we have to state local soundness differently depending on the label of the control-program step:
\begin{thm}[Local Soundness]
  \label{thm:local-sound}
  All of the following are true:
  \begin{itemize}
  \item If $\ContStepsTo{l}{\ProjectC{C_1}{\GenericLoc}}{E}$ where $\LabelMerge{l}{l}{\SysTau}$, then there are $R$ and $C_2$ such that $\ProjectR{R}{\GenericLoc} = l$, $\ProjectC{C_2}{\GenericLoc} = E$, and $\ChorStepsTo{R}{\emptyset}{C_1}{C_2}$.
  \item If $\ContStepsTo{l_1}{\ProjectC{C_1}{\GenericLocI}}{E_1}$ and $\ContStepsTo{l_2}{\ProjectC{C_2}{\GenericLocII}}{E_2}$ where $\LabelMerge{l_1}{l_2}{\CommLabel{\GenericLocI}{v}{\GenericLocII}}$, then there are $R$ and $C_2$ such that
    (a)~$\ProjectC{C_2}{\GenericLocI} = E_1$,
    (b)~$\ProjectC{C_2}{\GenericLocII} = E_2$,
    (c)~$\ProjectR{R}{\GenericLocI} = l_1$,
    (d)~$\ProjectR{R}{\GenericLocII} = l_2$, and
    (e)~$\ChorStepsTo{R}{\emptyset}{C_1}{C_2}$.
  \item If $\ContStepsTo{l_1}{\ProjectC{C_1}{\GenericLocI}}{E_1}$ and $\ContStepsTo{l_2}{\ProjectC{C_2}{\GenericLocII}}{E_2}$ where $\LabelMerge{l_1}{l_2}{\ChoiceLabel{\GenericLocI}{\GenChoice}{\GenericLocII}}$, then there are $R$ and $C_2$ such that
    (a)~$\ProjectC{C_2}{\GenericLocI} = E_1$,
    (b)~$\ProjectC{C_2}{\GenericLocII} = E_2$,
    (c)~$\ProjectR{R}{\GenericLocI} = l_1$,
    (d)~$\ProjectR{R}{\GenericLocII} = l_2$, and
    (e)~$\ChorStepsTo{R}{\emptyset}{C_1}{C_2}$.
  \item If $\LocationNames(C_1) \subseteq \mathfrak{L} \neq \emptyset$, $\LabelMerge{l}{l}{\SysSyncTau}$, and for every $\GenericLoc \in \mathfrak{L}$ $\ContStepsTo{l}{\ProjectC{C_1}{\GenericLoc}}{E_{\GenericLoc}}$, then there are $R$ and $C_2$ such that
    (a)~for every $\GenericLoc \in \mathfrak{L}$, $\ProjectR{R}{\GenericLoc} = l_{\GenericLoc}$, and
    (b)~$\ChorStepsTo{R}{\emptyset}{C_1}{C_2}$
  \end{itemize}
\end{thm}

The requirement that $\LocationNames(C_1) \subseteq \mathfrak{L} \neq \emptyset$ may seem unusual in the last case of Theorem~\ref{thm:local-sound}.
In order to $\beta$-reduce a function call or a subexpression $\DefLocal{\GenericLoc}{x}{C_1}{C_2}$ inside a choreography~$C$, every location in $C$ must be able to perform the same $\beta$-reduction.
However, if we only require that every location in $\LocationNames(C)$ be able to make a step, the requirement might be trivial if $C$ does not name any locations.
Therefore, we require that every location in some nonempty set of locations $\mathfrak{L}$ which contains every location in $\LocationNames(C)$ be able to make a step.
In the common case where $\LocationNames(C)$ is nonempty, this restriction is the same as the simpler requirement.

Lifting soundness from control~programs to systems yields a much simpler theorem.
However, note how the strange requirement for $\beta$-reduction steps infects the theorem:
\begin{thm}[Global Soundness]
  \label{thm:global-sound}
  If $\LocationNames(C_1) \subseteq \mathfrak{L} \neq \emptyset$ and $\SysStepsTo{L}{\ProjectC{C_1}{\mathfrak{L}}}{\Pi}$ then there is an $R$ and $C_2$ such that
  (a)~$\Defed{\ProjectC{C_2}{\mathfrak{L}}}$,
  (b)~$\LND{\ProjectC{C_2}{\mathfrak{L}}}{\Pi}$,
  (c)~$\CompileR{R} = L$, and
  (d)~$\ChorStepsTo{R}{\emptyset}{C_1}{C_2}$.
\end{thm}

We use these to develop the most important theorem for Pirouette (and choreographies in general): deadlock freedom by design.
Interestingly, the proof is a simple interplay of soundness and completeness for our translation along with type soundness.
\begin{thm}[Deadlock Freedom By Design]
  \label{thm:deadlock-freedom}
  If choreography typing enjoys both progress and preservation (e.g., if local typing is sound), $\ChorTyping{\cdot}{\cdot}{C_1}{\ChorTypeTau}$, and $\SysStepsManyTo{Ls}{\ProjectC{C_1}{\mathfrak{L}}}{\Pi}$, then either every location in $\Pi$ maps to a control-language value, or there are $L$ and $\Pi'$ such that $\SysStepsTo{L}{\Pi}{\Pi'}$.
\end{thm}
\begin{proof}[Sketch of the mechanized proof]
  By soundness, there are a list of redices~$Rs$ and a choreography~$C_2$ such that $\ChorStepsTo{Rs}{\emptyset}{C_1}{C_2}$ and $\LND{\ProjectC{C_2}{\mathfrak{L}}}{\Pi}$.
  By preservation, $\ChorTyping{\cdot}{\cdot}{C_2}{\ChorTypeTau}$, and by progress either $C_2$ is a choreography~value or there are $R$ and $C_3$ such that $\ChorStepsTo{R}{\emptyset}{C_2}{C_3}$.
  If $C_2$ is a value, then every location in $\Pi$ maps to a control-language value, and we are done.
  Otherwise, completeness tells us there is a $\Pi''$ such that $\SysStepsTo{\CompileR{R}}{\ProjectC{C_2}{\mathfrak{L}}}{\Pi''}$.
  We can then use Theorem~\ref{thm:lift-lower-global} to get a $\Pi'$ such that $\SysStepsTo{\CompileR{R}}{\Pi}{\Pi'}$, as desired.
\end{proof}


%% file: coq.tex
\section{Notes on the Coq Code}
\label{sec:notes-coq-code}

In this paper, we have presented all of our work in a standard, named style.
However, in our Coq code we use a nameless style with de~Bruijn indices.
While normally this would not pose any difficulties, we have taken the unusual step of treating the local~language generically.
Here, we give the nameless version of our requirements on the local language, ensuring that the transition to the named style does not create undue confusion.

First, we require that expressions have decidable \emph{syntactic} equality---usually a trivial requirement, since programs are usually sentences from a context-free grammar.
Next, we formalize the fact that variables must be expressions by requiring a mathematical function $\text{var}$ from natural numbers to expressions.
In all of our examples, this is a terminal in the language, but this is not required.
Instead of the ability to compute the set of free variables of an expression, we require a predicate $\ExprClosed[n]{e}$, which means that there are no free variables above~$n$ in~$e$.
We require that $\ExprClosed[n]{\text{var}(m)}$ if and only if $m < n$, and we write $\ExprClosed{e}$ for $\ExprClosed[0]{e}$. 

Substitution is changed to allow for infinite parallel substitution.
Formally, we require an operation $\Subst{e}{\sigma}$ where $\sigma$ is a mathematical function from natural numbers to expressions.
This must obey the following equations:
\begin{itemize}
\item $\Subst{\text{var}(n)}{\sigma} = \sigma(n)$
\item $\Subst{e}{n \mapsto \text{var}(n)} = e$
\item $\Subst{(\Subst{e}{\sigma_1})}{\sigma_2} = \Subst{e}{n \mapsto (\Subst{\sigma_1(n)}{\sigma_2})}$
\item if $\sigma_1(n) = \sigma_2(n)$ for every natural number~$n$, then $\Subst{e}{\sigma_1} = \Subst{e}{\sigma_2}$    
\end{itemize}
The last requirement is important because Coq is an \emph{intensional} type theory, so functions that behave the same on all inputs are not necessarily equal.
The other three equations are nameless versions of the equations in Section~\ref{sec:system-model}.

Finally, we require an additional \emph{renaming} operation, which we write $e \langle\xi\rangle$, where $\xi$ is a mathematical function from natural numbers to natural numbers.
This must satisfy the equation $e \langle\xi\rangle = \Subst{e}{n \mapsto \text{var}(\xi(n))}$.
While this equation could serve as a definition of renaming, it is often useful to define and reason about it separately.
For instance, we can use renaming to define exchange and weakening in typing.

Typing judgments in the nameless setting use contexts which are mathematical functions from natural numbers to types.
The requirements in Figure~\ref{fig:local-type-rules} are then equivalent to the following requirements:
\begin{mathpar}
  \infer[Var]{ }{\ExprTyping{\Gamma}{\text{var}(n)}{\Gamma(n)}} \and
  \infer[Extensionality]{\forall n, \Gamma(n) = \Delta(n)\\ \ExprTyping{\Gamma}{e}{\ExprTypeTau}}{\ExprTyping{\Delta}{e}{\ExprTypeTau}} \and
  \infer[Weakening]{\ExprTyping{\Gamma}{e}{\ExprTypeTau}}{\ExprTyping{\Gamma \circ \xi}{e \langle\xi\rangle}{\ExprTypeTau}} \and
  \infer[Strengthening]{\ExprClosed[n]{e}\\\forall m < n, \Gamma(m) = \Delta(m)\\ \ExprTyping{\Gamma}{e}{\ExprTypeTau}}{\ExprTyping{\Delta}{e}{\ExprTypeTau}}\and
  \infer[Substitution]{\forall n, \ExprTyping{\Gamma}{\sigma(n)}{\Delta(n)}\\ \ExprTyping{\Delta}{e}{\ExprTypeTau}}{\ExprTyping{\Gamma}{\Subst{e}{\sigma}}{\ExprTypeTau}}
\end{mathpar}
Note that both \textsc{Exchange} and \textsc{Weakening} are absorbed into the more-general \textsc{Weakening} rule.
Moreover, we made \textsc{Substitution} infinitary.
The \textsc{Var} and \textsc{Strengthening} rules are straightforward transformations of the named versions in Figure~\ref{fig:local-type-rules}.


%% file: related.tex
\section{Related Work}
\label{sec:related-work}

\subsection{Choreographies}
\label{sec:related-choreographies}

Choreographies originate as a way of writing web services.
The W3C released a report on choreographies as a way to write web services~\citep{W3C04}, which originated the idea of endpoint projection~\citep{ZongyanXCH07} (see \citet{Cruz-FilipeM17} for the historical note).
Soon, there was work formalizing and understanding choreographies from this point of view.

Choreographic programming as a paradigm originates with Montesi's Ph.D. thesis~\citep{Montesi13}.
There, he develops the first choreographic programming language and begins exploring the formal properties thereof.
Later works extend this idea, with \citet{Cruz-FilipeM17} creating a core calculus of choreographies and \citet{CarboneM13} combining choreographies with session types to ensure that endpoint projection never fails.
However these works were all lower-order, used equivalence-based semantics, and had no mechanized metatheory.

More closely related with this project, \citet{Cruz-FilipeM17c} create a choreographic language with procedures and the ability to make procedure calls.
While parts of our endpoint projection definition are inspired by this work, procedural choreographies remain resolutely lower-order.
Procedures are global and shared, but cannot be treated as data; moreover, they cannot take other choreographies as inputs nor return them as outputs.
However, unlike our functions, these procedures are able to take locations as inputs.
This is harder in our setting because locations are part of types; we consider polymorphism (including location polymorphism) to be future work.

Some work has been done on higher-order choreographies.
The Choral Project~\cite{GiallorenzoMP20} builds choreographies on top of object-oriented programming.
However, Choral has never been formalized, and so the theoretical underpinnings of higher-order choreographies were left unexplored.

Significantly, concurrent and completely independent work on functional choreographies has recently been undertaken by \citet{CruzFilipeGLMP21}.
This work resembles Pirouette in several ways, but has a significantly different design.
Most importantly, they identify the local language with the choreography language.
This design choice allows functions to be sent as messages, but they require that messages only mention a single location, preventing functions sent as messages from including communication.
It also leads to a simpler system, but makes it much harder for them to identify the core features of the local language which allow for type soundness of the choreographic language.
Moreover, their operational semantics allows for functions to be reduced; this allows them to not need as many synchronizations as Pirouette projection requires.
However, it violates the common requirement of call-by-value languages that functions be treated as values and makes execution more unpredictable.
We preferred to keep the simpler and more-traditional operational semantics.
Finally, they do not have any mechanized formalization of their system.

\paragraph{Mechanization}
Two very recent papers have mechanized the metatheory of core choreographies, but both are lower-order.
\citet{CruzFilipeMP21a} use Coq to formalize the core calculus of choreographies developed by \citet{Cruz-FilipeM17}, modified so that the language of local computations is a parameter.
The focus is on proving that the language is Turing-complete.
In the meantime, \citet{CruzFilipeMP21b} extend that work to verify (in Coq) endpoint projection and deadlock freedom of the same language.

Interestingly, both these papers formalize out-of-order execution for choreographies using a labeled-transition system, similar to ours.
As far as we are aware, they, along with a forthcoming book on choreographies~\cite{Montesi21} were the first to do so, though our work was formalized before we discovered their concurrent work.
Their work takes inspiration from \citet{HondaYC16}, who use a similar labeled-transition system semantics for multiparty session types.
Thus, they do not use any construct similar to our block sets.
This means that they cannot allow locations to locally reduce their messages when the receiver of that message is blocked. 
While for their language this does not matter (since messages do not reduce at all), in our setting this would be problematic.

\subsection{Functional Concurrent Programming}
\label{sec:funct-conc-progr}

There is a long tradition of mixing functional and concurrent programming in principled ways.
In practice, this often leads to languages that look a lot like our control~language, including the concurrency features in Racket~\citep{ConcurrentRacket}.
The first academic language with channels and communication in a functional language was Facile~\citep{GiacaloneMP89}.
Even before then, work on parallelizing compilers for functional programming was popular.
\citet{Burton87} considered adding annotations to functions describing when arguments should be evaluated at what location.

Currently, the most related academic project might be Links~\citep{CooperLWY06} and its core language, the RPC calculus~\citep{CooperW09}.
Links also provides for a mixture of higher-order typed functional programming and concurrency.
However, in Links, communication always happens at function boundaries, unlike Pirouette.
Moreover, the RPC calculus only allows for one thread of execution, even when multiple unrelated locations can profitably be taking actions.
This is in contradiction to the work on choreographies---and Pirouette in particular---where unrelated locations can compute and even pass messages concurrently.

Another language which mixes functional and concurrent programming is \citeauthor{MurphyCH07}'s ML5~\citep{LicataH10,MurphyCH07}.
ML5 has a type system for communication based on the Kripke semantics of modal logic.
Interestingly, our type system also takes inspiration from modal logic: our type system can be seen as (the Curry-Howard analogue of) a degenerate form of the \emph{proof system} for modal logic.
Their use of the Kripke semantics makes sense, however, because different worlds can be viewed as different machines.
Unfortunately, this led to difficulties in shipping some data that could be treated as code, so ML5 focuses on annotating \emph{mobile code}, which includes static types like strings and integers, but not functions or local resources like arrays.
This restriction is necessary due to their use of local state.
Since Pirouette has no state, is has no such constraint.


%% file: appendices.tex
\input{chor_sem}
\input{chor_types}
\input{equiv_full}
\input{cont_sem}
\input{cont_merge}
\input{lnd}
\input{thms}

%% file: chor_sem.tex
\section{Full Pirouette Operational Semantics}
\label{sec:full-chor-oper}

\subsection{Local Substitution, Defined}
\label{sec:local-subst-defin}

\begin{mathparpagebreakable}
  \Subst{X}{\SingleChorExprSubst{\GenericLoc}{x}{e}} = X \and
  \Subst{(\Own{\GenericLocI}{e_1})}{\SingleChorExprSubst{\GenericLocII}{x}{e_2}} =
  \left\{
    \begin{array}{ll}
      \Own{\GenericLocI}{\Subst{e_1}{\SingleSubst{x}{e_2}}} & \text{if}~\GenericLocI = \GenericLocII\\
      \Own{\GenericLocI}{e_1} & \text{otherwise}
    \end{array}
  \right. \and
  \Subst{(\Send{\GenericLocI}{e_1}{\GenericLocII}{x}{C})}{\SingleChorExprSubst{\GenericLocIII}{y}{e}} =
  \left\{
    \begin{array}{ll}
      \Send{\GenericLocI}{\Subst{e_1}{\SingleSubst{x}{e_2}}}{\GenericLocII}{y}{C}
      & \begin{array}{l}\text{if}~\GenericLocI = \GenericLocII = \GenericLocIII\\\text{and}~x=y\end{array}\\[1em]
      \Send*{\GenericLocI}{\Subst{e_1}{\SingleSubst{y}{e_2}}}{\GenericLocII}{x}{~(\Subst{C}{\SingleChorExprSubst{\GenericLocIII}{y}{e}})}
      & \begin{array}{l}\text{if}~\GenericLocI = \GenericLocIII\\\text{and either}\begin{array}[t]{l}\GenericLocII \neq \GenericLocIII\\\text{or}~x\neq y\end{array}\end{array}\\
      \Send{\GenericLocI}{e_1}{\GenericLocII}{x}{C} & \begin{array}{l}\text{if}~\GenericLocI \neq \GenericLocIII\\\text{and}~\GenericLocII = \GenericLocIII\\\text{and}~x = y\end{array}\\[2em]
      \Send*{\GenericLocI}{e_1}{\GenericLocII}{x}{~(\Subst{C}{\SingleChorExprSubst{\GenericLocIII}{y}{e}})} & \begin{array}{l}\text{if}~\GenericLocI \neq \GenericLocIII\\\text{and either}\begin{array}[t]{l}\GenericLocII \neq \GenericLocIII\\\text{or}~x \neq y\end{array}\end{array}
    \end{array}
  \right.\and
  \Subst{(\ChorIf{\GenericLocI}{e}{C_1}{C_2})}{\SingleChorExprSubst{\GenericLocII}{x}{e}} =
  \left\{
    \begin{array}{ll}
      \ChorIf*{\GenericLocI}{(\Subst{e}{\SingleSubst{x}{e}})}{(\Subst{C_1}{\SingleChorExprSubst{\GenericLocII}{x}{e}})}{(\Subst{C_2}{\SingleChorExprSubst{\GenericLocII}{x}{e}})} & \text{if}~\GenericLocI = \GenericLocII\\[2em]
      \ChorIf*{\GenericLocI}{e}{(\Subst{C_1}{\SingleChorExprSubst{\GenericLocII}{x}{e}})}{(\Subst{C_2}{\SingleChorExprSubst{\GenericLocII}{x}{e}})} & \text{otherwise}
    \end{array}
  \right. \and
  \Subst{(\Sync{\GenericLocI}{\GenChoice}{\GenericLocII}{C})}{\SingleChorExprSubst{\GenericLocIII}{x}{e}} = \Sync{\GenericLocI}{\GenChoice}{\GenericLocII}{(\Subst{C}{\SingleChorExprSubst{\GenericLocIII}{x}{e}})} \\
  \Subst{(\DefLocal{\GenericLocI}{x}{C_1}{C_2})}{\SingleChorExprSubst{\GenericLocII}{y}{e}} =
  \left\{
    \begin{array}{ll}
      \DefLocal*{\GenericLocI}{x}{(\Subst{C_1}{\SingleChorExprSubst{\GenericLocII}{y}{e}})}{C_2} & \begin{array}{l}\text{if}~\GenericLocI = \GenericLocII\\\text{and}~x=y\end{array}\\[1em]
      \DefLocal*{\GenericLocI}{x}{(\Subst{C_1}{\SingleChorExprSubst{\GenericLocII}{y}{e}})}{(\Subst{C_2}{\SingleChorExprSubst{\GenericLocII}{y}{e}})} & \text{otherwise}
    \end{array}
  \right. \and
  \Subst{(\FunLocal{\GenericLocI}{F}{X}{C})}{\SingleChorExprSubst{\GenericLocII}{y}{e}} =
  \left\{
    \begin{array}{ll}
      \FunLocal{\GenericLocI}{F}{X}{C} & \text{if}~\GenericLocI=\GenericLocII~\text{and}~x=y\\
      \FunLocal{\GenericLocI}{F}{X}{(\Subst{C}{\SingleChorExprSubst{\GenericLocII}{y}{e}})} & \text{otherwise}
    \end{array}
  \right. \and
  \Subst{(\AppLocal{\GenericLocI}{C}{e_1})}{\SingleChorExprSubst{\GenericLocII}{x}{e_2}} =
  \left\{
    \begin{array}{ll}
      \AppLocal{\GenericLocI}{(\Subst{C}{\SingleChorExprSubst{\GenericLocII}{x}{e_2}})}{(\Subst{e_1}{\SingleSubst{x}{e_2}})} & \text{if}~\GenericLocI = \GenericLocII\\
      \AppLocal{\GenericLocI}{(\Subst{C}{\SingleChorExprSubst{\GenericLocII}{x}{e_2}})}{e_1} & \text{otherwise}
    \end{array}
  \right. \and
  \Subst{(\FunGlobal{F}{X}{C})}{\SingleChorExprSubst{\GenericLoc}{x}{e}} = \FunGlobal{F}{X}{(\Subst{C}{\SingleChorExprSubst{\GenericLoc}{x}{e}})} \and
  \Subst{(\AppGlobal{C_1}{C_2})} = \AppGlobal{(\Subst{C_1}{\SingleChorExprSubst{\GenericLoc}{x}{e}})}{(\Subst{C_2}{\SingleChorExprSubst{\GenericLoc}{x}{e}})}
\end{mathparpagebreakable}

\subsection{Block-Set Semantics}
\label{sec:block-set-semantics}
See Section~\ref{sec:oper-semant} for discussion.

  \begin{syntax}
    \category[Redices]{R} \alternative{\RDone{\GenericLoc}{e_1}{e_2}}
    \alternative{\RIfE{\GenericLoc}{e_1}{e_2}}
    \alternative{\RIfT{\GenericLoc}}
    \alternative{\RIfF{\GenericLoc}}
    \alternativeLine{\RSendE{\GenericLoc}{e_1}{e_2}{\GenericLocII}}
    \alternative{\RSendV{\GenericLoc}{v}{\GenericLocII}}
    \alternative{\RSync{\GenericLoc}{\GenChoice}{\GenericLocII}}
    \alternativeLine{\RDefLocal{\GenericLoc}{v}}
    \alternative{\RAppLocalE{\GenericLoc}{e_1}{e_2}}
    \alternative{\RAppLocal{\GenericLoc}{v}}
    \alternativeLine{\RAppGlobal}
    \alternative{\RArg(R)} \alternative{\RFun(R)}
  \end{syntax}\\

\begin{mathparpagebreakable}
  \infer[DoneE]{\GenericLoc \notin B\\ \ExprStepsTo{e_1}{e_2}}{\ChorStepsTo{\RDone{\GenericLoc}{e_1}{e_2}}{B}{\Own{\GenericLoc}{e_1}}{\Own{\GenericLoc}{e_2}}} \\
  \infer[SendE]{\GenericLocI \notin B\\ \GenericLocI \neq \GenericLocII\\ \ExprStepsTo{e_1}{e_2}}{
    \ChorStepsTo{\RSendE{\GenericLocI}{e_1}{e_2}{\GenericLocII}}{B}{\Send{\GenericLocI}{e_1}{\GenericLocII}{x}{C}}{\Send{\GenericLocI}{e_2}{\GenericLocII}{x}{C}}} \and
  \infer[SendI]{\ChorStepsTo{R}{B \mathrel{\cup} \{\GenericLocI, \GenericLocII\}}{C_1}{C_2}}{
    \ChorStepsTo{R}{B}{\Send{\GenericLocI}{e}{\GenericLocII}{x}{C_1}}{\Send{\GenericLocI}{e}{\GenericLocII}{x}{C_2}}}\and
  \infer[SendV]{\GenericLocI \notin B\\ \GenericLocII \notin B\\ \ExprValue{v}\\ \GenericLocI \neq \GenericLocII}{
    \ChorStepsTo{\RSendV{\GenericLocI}{v}{\GenericLocII}}{B}{\Send{\GenericLocI}{v}{\GenericLocII}{x}{C}}{\Subst{C}{\SingleChorExprSubst{\GenericLocII}{x}{v}}}}\\
  \infer[IfE]{
    \GenericLoc \notin B\\
    \ExprStepsTo{e_1}{e_2}
  }
  {
    \ChorStepsTo{\RIfE{\GenericLoc}{e_1}{e_2}}{B}{\ChorIf{\GenericLoc}{e_1}{C_1}{C_2}}{\ChorIf{\GenericLoc}{e_2}{C_1}{C_2}}
  } \and
  \infer[IfI]{
    \ChorStepsTo{R}{B \cup \{\GenericLoc\}}{C_1}{C_1'}\\
    \ChorStepsTo{R}{B \cup \{\GenericLoc\}}{C_2}{C_2'}
  }
  {
    \ChorStepsTo{R}{B}{\ChorIf{\GenericLoc}{e}{C_1}{C_2}}{\ChorIf{\GenericLoc}{e}{C_1'}{C_2'}}
  } \and
  \infer[IfT]{
    \GenericLoc \notin B\\
  }
  {
    \ChorStepsTo{\RIfT{\GenericLoc}}{B}{\ChorIf{\GenericLoc}{\True}{C_1}{C_2}}{C_1}
  } \and
  \infer[IfF]{
    \GenericLoc \notin B\\
  }
  {
    \ChorStepsTo{\RIfT{\GenericLoc}}{B}{\ChorIf{\GenericLoc}{\False}{C_1}{C_2}}{C_2}
  }\\
  \infer[DefLocalI]{
    \ChorStepsTo{R}{B}{C_1}{C_1'}
  }{
    \ChorStepsTo{\RArg(R)}{B}{\DefLocal{\GenericLoc}{x}{C_1}{C_2}}{\DefLocal{\GenericLoc}{x}{C_1'}{C_2}}
  } \and
  \infer[DefLocal]{
    \ExprValue{v}
  }{
    \ChorStepsTo{\RDefLocal{\GenericLoc}{v}}{B}{\DefLocal{\GenericLoc}{x}{\Own{\GenericLoc}{v}}{C}}{\Subst{C}{\SingleChorExprSubst{\GenericLoc}{x}{v}}}
  }\and
  \infer[AppLocalFun]{
    \ChorStepsTo{R}{B}{C_1}{C_2}
  }{
    \ChorStepsTo{\RFun(R)}{B}{\AppLocal{\GenericLoc}{C_1}{e}}{\AppLocal{\GenericLoc}{C_2}{e}}
  } \and
  \infer[AppLocalArg]{
    \GenericLoc \notin B\\
    \ExprStepsTo{e_1}{e_2}
  }{
    \ChorStepsTo{\RAppLocalE{\GenericLoc}{e_1}{e_2}}{B}{\AppLocal{\GenericLoc}{C}{e_1}}{\AppLocal{\GenericLoc}{C}{e_2}}
  } \and
  \infer[AppLocal]{
    \ExprValue{v}
  }{
    \ChorStepsTo{\RAppLocal{\GenericLoc}{v}}{\emptyset}{\AppLocal{\GenericLoc}{(\FunLocal{\GenericLoc}{F}{x}{C})}{e}}{\Subst{\Subst{C}{\SingleChorExprSubst{\GenericLoc}{x}{v}}}{\SingleSubst{F}{\FunLocal{\GenericLoc}{F}{x}{C}}}}
  }\\
  \infer[AppGlobalFun]{
    \ChorStepsTo{R}{B}{C_1}{C_1'}
  }{
    \ChorStepsTo{\RFun(R)}{B}{\AppGlobal{C_1}{C_2}}{\AppGlobal{C_1'}{C_2}}
  }\and
  \infer[AppGlobalArg]{
    \ChorStepsTo{R}{B}{C_2}{C_2'}
  }{
    \ChorStepsTo{\RArg(R)}{B}{\AppGlobal{C_1}{C_2}}{\AppGlobal{C_1}{C_2'}}
  } \and
  \infer[AppGlobal]{
    \ChorValue(V)
  }{
    \ChorStepsTo{\RAppGlobal}{\emptyset}{\AppGlobal{(\FunGlobal{F}{X}{C})}{V}}{\Subst{C}{\SingleSubst{X}{V}, \SingleSubst{F}{\FunGlobal{F}{X}{C}}}}
  }\\
  \infer[SyncI]{
    \ChorStepsTo{R}{B \cup \{\GenericLocI, \GenericLocII\}}{C_1}{C_2}
  }{
    \ChorStepsTo{R}{B}{\Sync{\GenericLocI}{\GenChoice}{\GenericLocII}{C_1}}{\Sync{\GenericLocI}{\GenChoice}{\GenericLocII}{C_2}}
  } \and
  \infer[Sync]{
    \GenericLocI \notin B\\ \GenericLocII \notin B\\
    \GenericLocI \neq \GenericLocII
  }{
    \ChorStepsTo{R}{B}{\Sync{\GenericLocI}{\GenChoice}{\GenericLocII}{C}}{C}
  }
\end{mathparpagebreakable}

\subsection{Weak Block-Set Semantics}
\label{sec:weak-block-set}

See Section~\ref{sec:equational-reasoning} for discussion.

\begin{mathparpagebreakable}
  \infer[DoneE]{\ExprStepsTo{e_1}{e_2}}{\WeakStepsTo{\RDone{\GenericLoc}{e_1}{e_2}}{\emptyset}{\Own{\GenericLoc}{e_1}}{\Own{\GenericLoc}{e_2}}} \\
  \infer[SendE]{\GenericLocI \notin B\\ \GenericLocII \notin B\\ \GenericLocI \neq \GenericLocII\\ \ExprStepsTo{e_1}{e_2}}{
    \WeakStepsTo{\RSendE{\GenericLocI}{e_1}{e_2}{\GenericLocII}}{B}{\Send{\GenericLocI}{e_1}{\GenericLocII}{x}{C}}{\Send{\GenericLocI}{e_2}{\GenericLocII}{x}{C}}} \and
  \infer[SendI]{\WeakStepsTo{R}{B \mathrel{\cup} \{\GenericLocI, \GenericLocII\}}{C_1}{C_2}}{
    \WeakStepsTo{R}{B}{\Send{\GenericLocI}{e}{\GenericLocII}{x}{C_1}}{\Send{\GenericLocI}{e}{\GenericLocII}{x}{C_2}}}\and
  \infer[SendV]{\GenericLocI \notin B\\ \GenericLocII \notin B\\ \ExprValue{v}\\ \GenericLocI \neq \GenericLocII}{
    \WeakStepsTo{\RSendV{\GenericLocI}{v}{\GenericLocII}}{B}{\Send{\GenericLocI}{v}{\GenericLocII}{x}{C}}{\Subst{C}{\SingleChorExprSubst{\GenericLocII}{x}{v}}}}\\
  \infer[IfE]{
    \GenericLoc \notin B\\
    \ExprStepsTo{e_1}{e_2}
  }
  {
    \WeakStepsTo{\RIfE{\GenericLoc}{e_1}{e_2}}{B}{\ChorIf{\GenericLoc}{e_1}{C_1}{C_2}}{\ChorIf{\GenericLoc}{e_2}{C_1}{C_2}}
  } \and
  \infer[IfI]{
    \WeakStepsTo{R}{B \cup \{\GenericLoc\}}{C_1}{C_1'}\\
    \WeakStepsTo{R}{B \cup \{\GenericLoc\}}{C_2}{C_2'}
  }
  {
    \WeakStepsTo{R}{B}{\ChorIf{\GenericLoc}{e}{C_1}{C_2}}{\ChorIf{\GenericLoc}{e}{C_1'}{C_2'}}
  } \and
  \infer[IfT]{
    \GenericLoc \notin B\\
  }
  {
    \WeakStepsTo{\RIfT{\GenericLoc}}{B}{\ChorIf{\GenericLoc}{\True}{C_1}{C_2}}{C_1}
  } \and
  \infer[IfF]{
    \GenericLoc \notin B\\
  }
  {
    \WeakStepsTo{\RIfT{\GenericLoc}}{B}{\ChorIf{\GenericLoc}{\False}{C_1}{C_2}}{C_2}
  }\\
  \infer[DefLocalI]{
    \WeakStepsTo{R}{\emptyset}{C_1}{C_1'}
  }{
    \WeakStepsTo{\RArg(R)}{\emptyset}{\DefLocal{\GenericLoc}{x}{C_1}{C_2}}{\DefLocal{\GenericLoc}{x}{C_1'}{C_2}}
  } \and
  \infer[DefLocal]{
    \ExprValue{v}
  }{
    \WeakStepsTo{\RDefLocal{\GenericLoc}{v}}{\emptyset}{\DefLocal{\GenericLoc}{x}{\Own{\GenericLoc}{v}}{C}}{\ChorSubst{C}{\SingleChorExprSubst{\GenericLoc}{x}{v}}}
  }\and
  \infer[AppLocalFun]{
    \WeakStepsTo{R}{\emptyset}{C_1}{C_2}
  }{
    \WeakStepsTo{\RFun(R)}{\emptyset}{\AppLocal{\GenericLoc}{C_1}{e}}{\AppLocal{\GenericLoc}{C_2}{e}}
  } \and
  \infer[AppLocalArg]{
    \ExprStepsTo{e_1}{e_2}
  }{
    \WeakStepsTo{\RAppLocalE{\GenericLoc}{e_1}{e_2}}{\emptyset}{\AppLocal{\GenericLoc}{C}{e_1}}{\AppLocal{\GenericLoc}{C}{e_2}}
  } \and
  \infer[AppLocal]{
    \ExprValue{v}
  }{
    \WeakStepsTo{\RAppLocal{\GenericLoc}{v}}{\emptyset}{\AppLocal{\GenericLoc}{(\FunLocal{\GenericLoc}{F}{x}{C})}{e}}{\Subst{\Subst{C}{\SingleChorExprSubst{\GenericLoc}{x}{v}}}{\SingleSubst{F}{\FunLocal{\GenericLoc}{F}{x}{C}}}}
  }\\
  \infer[AppGlobalFun]{
    \WeakStepsTo{R}{\emptyset}{C_1}{C_1'}
  }{
    \WeakStepsTo{\RFun(R)}{\emptyset}{\AppGlobal{C_1}{C_2}}{\AppGlobal{C_1'}{C_2}}
  }\and
  \infer[AppGlobalArg]{
    \WeakStepsTo{R}{\emptyset}{C_2}{C_2'}
  }{
    \WeakStepsTo{\RArg(R)}{\emptyset}{\AppGlobal{C_1}{C_2}}{\AppGlobal{C_1}{C_2'}}
  } \and
  \infer[AppGlobal]{
    \ChorValue(V)
  }{
    \WeakStepsTo{\RAppGlobal}{\emptyset}{\AppGlobal{(\FunGlobal{F}{X}{C})}{V}}{\Subst{C}{\SingleSubst{X}{V}, \SingleSubst{F}{\FunGlobal{F}{X}{C}}}}
  }\\
  \infer[SyncI]{
    \WeakStepsTo{R}{B \cup \{\GenericLocI, \GenericLocII\}}{C_1}{C_2}
  }{
    \WeakStepsTo{R}{B}{\Sync{\GenericLocI}{\GenChoice}{\GenericLocII}{C_1}}{\Sync{\GenericLocI}{\GenChoice}{\GenericLocII}{C_2}}
  } \and
  \infer[Sync]{
    \GenericLocI \notin B\\ \GenericLocII \notin B\\
    \GenericLocI \neq \GenericLocII
  }{
    \WeakStepsTo{R}{B}{\Sync{\GenericLocI}{\GenChoice}{\GenericLocII}{C}}{C}
  }
\end{mathparpagebreakable}

\subsection{Equivalence-Based Semantics}
\label{sec:equiv-based-semant}

See Section~\ref{sec:equational-reasoning} for discussion.

\begin{mathparpagebreakable}
  \infer[EquivStep]{
    \ChorEquiv{C_1}{C_1'}\\
    \EquivStepsTo{C_1'}{C_2'}\\
    \ChorEquiv{C_2'}{C_2}
  }{
    \EquivStepsTo{C_1}{C_2}
  } \and
  \infer[DoneE]{
    \ExprStepsTo{e_1}{e_2}
  }{
    \EquivStepsTo{\Own{\GenericLoc}{e_1}}{\Own{\GenericLoc}{e_2}}
  } \and
  \infer[SendE]{
    \ExprStepsTo{e_1}{e_2}\\
    \GenericLocI \neq \GenericLocII
  }{
    \EquivStepsTo{\Send{\GenericLocI}{e_1}{\GenericLocII}{x}{C}}{\Send{\GenericLocI}{e_2}{\GenericLocII}{x}{C}}
  }\and
  \infer[SendV]{
    \ExprValue{v}\\
    \GenericLocI \neq \GenericLocII
  }{
    \EquivStepsTo{\Send{\GenericLocI}{v}{\GenericLocII}{x}{C}}{\Subst{C}{\SingleChorExprSubst{\GenericLocII}{x}{v}}}
  } \and
  \infer[IfE]{
    \ExprStepsTo{e_1}{e_2}
  }{
    \EquivStepsTo{\ChorIf{\GenericLoc}{e_1}{C_1}{C_2}}{\ChorIf{\GenericLoc}{e_2}{C_1}{C_2}}
  }\and
  \infer[IfTrue]{ 
  }{
    \EquivStepsTo{\ChorIf{\GenericLoc}{\True}{C_1}{C_2}}{C_1}
  } \and 
  \infer[IfFalse]{ 
  }{
    \EquivStepsTo{\ChorIf{\GenericLoc}{\False}{C_1}{C_2}}{C_2}
  } \and
  \infer[Sync]{
    \GenericLocI \neq \GenericLocII
  }{
    \EquivStepsTo{\Sync{\GenericLocI}{\GenChoice}{\GenericLocII}{C}}{C}
  } \and
  \infer[DefLocalArg]{
    \EquivStepsTo{C_1}{C_1'}
  }{
    \EquivStepsTo{\DefLocal{\GenericLoc}{x}{C_1}{C_2}}{\DefLocal{\GenericLoc}{x}{C_1'}{C_2}}
  } \and
  \infer[DefLocal]{
    \ExprValue{v}
  }{
    \EquivStepsTo{\DefLocal{\GenericLoc}{x}{\Own{\GenericLoc}{v}}{C}}{\Subst{C}{\SingleChorExprSubst{\GenericLoc}{x}{v}}}
  } \and
  \infer[AppLocalFun]{
    \EquivStepsTo{C_1}{C_2}
  }{
    \EquivStepsTo{\AppLocal{\GenericLoc}{C_1}{e}}{\AppLocal{\GenericLoc}{C_2}{e}}
  } \and
  \infer[AppLocalArg]{
    \ExprStepsTo{e_1}{e_2}
  }{
    \EquivStepsTo{\AppLocal{\GenericLoc}{C}{e_1}}{\AppLocal{\GenericLoc}{C}{e_2}}
  } \and
  \infer[AppLocal]{
    \ExprValue{v}
  }{
    \EquivStepsTo{\AppLocal{\GenericLoc}{\left(\FunLocal{\GenericLoc}{F}{x}{C}\right)}{v}}{\Subst{C}{\SingleChorExprSubst{\GenericLoc}{x}{v}}}
  } \and
  \infer[AppGlobalFun]{
    \EquivStepsTo{C_1}{C_1'}
  }{
    \EquivStepsTo{\AppGlobal{C_1}{C_2}}{\AppGlobal{C_1'}{C_2}}
  } \and
  \infer[AppGlobalArg]{
    \EquivStepsTo{C_2}{C_2'}
  }{
    \EquivStepsTo{\AppGlobal{C_1}{C_2}}{\AppGlobal{C_1}{C_2'}}
  } \and
  \infer[AppGlobal]{
    \ChorValue(V)
  }{
    \EquivStepsTo{\AppGlobal{\left(\FunGlobal{F}{X}{C}\right)}{V}}{\Subst{C}{\SingleSubst{X}{C}}}
  }
\end{mathparpagebreakable}


%% file: chor_types.tex
\section{Full Choreography Type System}
\label{sec:full-chor-types}

\begin{mathparpagebreakable}
  \infer[Done]{\ExprTyping{\ProjectCtxt{\Gamma}{\GenericLoc}}{e}{\ExprTypeTau}}{\ChorTyping{\Gamma}{\Delta}{\Own{\GenericLoc}{e}}{\OwnedType{\GenericLoc}{\ExprTypeTau}}} \and
  \infer[Var]{\ChorTypingPair{X}{\ChorTypeTau} \in \Delta}{\ChorTyping{\Gamma}{\Delta}{X}{\ChorTypeTau}} \\
  \infer[Send]{\ExprTyping{\ProjectCtxt{\Gamma}{\GenericLoc}}{e}{\ExprTypeTauI}\\
    \ChorTyping{\Gamma, \ChorTypingTriple{\GenericLocII}{x}{\ExprTypeTauI}}{\Delta}{C}{\ChorTypeTauII}\\
    \GenericLoc \neq \GenericLocII
  }
  {\ChorTyping{\Gamma}{\Delta}{\Send{\GenericLoc}{e}{\GenericLocII}{x}{C}}{\ChorTypeTauII}} \and
  \infer[Sync]{\ChorTyping{\Gamma}{\Delta}{C}{\ChorTypeTau}\\ \GenericLoc \neq \GenericLocII}{\ChorTyping{\Gamma}{\Delta}{\Sync{\GenericLoc}{\GenChoice}{\GenericLocII}{C}}{\ChorTypeTau}}\and
  \infer[If]{\ExprTyping{\ProjectCtxt{\Gamma}{\GenericLoc}}{e}{\Bool}\\
    \ChorTyping{\Gamma}{\Delta}{C_1}{\ChorTypeTau}\\
    \ChorTyping{\Gamma}{\Delta}{C_2}{\ChorTypeTau}}
  {\ChorTyping{\Gamma}{\Delta}{\ChorIf{\GenericLoc}{e}{C_1}{C_2}}{\ChorTypeTau}}\and
  \infer[DefLocal]{
    \ChorTyping{\Gamma}{\Delta}{C_1}{\OwnedType{\GenericLoc}{\ExprTypeTauI}}\\
    \ChorTyping{\Gamma, \ChorTypingTriple{\GenericLoc}{x}{\ExprTypeTauI}}{\Delta}{C_2}{\ChorTypeTauII}}
  {\ChorTyping{\Gamma}{\Delta}{\DefLocal{\GenericLoc}{x}{C_1}{C_2}}{\ChorTypeTauII}}\\
  \infer[FunLocal]{
    \ChorTyping{\Gamma, \ChorTypingTriple{\GenericLoc}{x}{\ExprTypeTauI}}{\Delta, \ChorTypingPair{F}{\LocalFunType{\GenericLoc}{\ExprTypeTauI}{\ChorTypeTauII}}}{C}{\ChorTypeTauII}
  }{
    \ChorTyping{\Gamma}{\Delta}{\FunLocal{\GenericLoc}{F}{x}{C}}{\LocalFunType{\GenericLoc}{\ExprTypeTauI}{\ChorTypeTauII}}
  }\and
  \infer[FunGlobal]{
    \ChorTyping{\Gamma}{\Delta, \ChorTypingPair{F}{\GlobalFunType{\ChorTypeTauI}{\ChorTypeTauII}}, \ChorTypingPair{X}{\ChorTypeTauI}}{C}{\ChorTypeTauII}
  }{
    \ChorTyping{\Gamma}{\Delta}{\FunGlobal{F}{X}{C}}{\GlobalFunType{\ChorTypeTauI}{\ChorTypeTauII}}
  }\and
  \infer[AppLocal]{
    \ExprTyping{\ProjectCtxt{\Gamma}{\GenericLoc}}{e}{\ExprTypeTauI}\\
    \ChorTyping{\Gamma}{\Delta}{C}{\LocalFunType{\GenericLoc}{\ExprTypeTauI}{\ChorTypeTauII}}
  }{
    \ChorTyping{\Gamma}{\Delta}{\AppLocal{\GenericLoc}{C}{e}}
  }\and
  \infer[AppGlobal]{
    \ChorTyping{\Gamma}{\Delta}{C_1}{\GlobalFunType{\ChorTypeTauI}{\ChorTypeTauII}}\\
    \ChorTyping{\Gamma}{\Delta}{C_2}{\ChorTypeTauI}
  }{
    \ChorTyping{\Gamma}{\Delta}{\AppGlobal{C_1}{C_2}}{\ChorTypeTauII}
  }    
\end{mathparpagebreakable}


%% file: equiv_full.tex
\section{Full Definition of Choreography Equivalence}
\label{sec:full-chor-equiv}

\begin{mathparpagebreakable}
  \infer[Trans]{\ChorEquiv{C_1}{C_2}\\\ChorEquiv{C_2}{C_3}}{\ChorEquiv{C_1}{C_3}}\\
  \infer[VarRefl]{ }{\ChorEquiv{X}{X}} \and
  \infer[DoneRefl]{ }{\ChorEquiv{\Own{\GenericLoc}{e}}{\Own{\GenericLoc}{e}}}\and
  \infer[SendCong]{\ChorEquiv{C_1}{C_2}}{\ChorEquiv{\Send{\GenericLocI}{e}{\GenericLocII}{x}{C_1}}{\Send{\GenericLocI}{e}{\GenericLocII}{x}{C_2}}} \and
  \infer[SyncCong]{\ChorEquiv{C_1}{C_2}}{\ChorEquiv{\Sync{\GenericLocI}{\GenChoice}{\GenericLocII}{C_1}}{\Sync{\GenericLocI}{\GenChoice}{\GenericLocII}{C_2}}} \and
  \infer[IfCong]{\ChorEquiv{C_{1,1}}{C_{1,2}}\\ \ChorEquiv{C_{2,1}}{C_{2,2}}}{\ChorEquiv{\ChorIf{\GenericLoc}{e}{C_{1,1}}{C_{2,1}}}{\ChorIf{\GenericLoc}{e}{C_{1,2}}{C_{2,2}}}} \and
  \infer[DefLocalCong]{\ChorEquiv{C_{1,1}}{C_{1,2}}\\ \ChorEquiv{C_{2,1}}{C_{2,2}}}{\ChorEquiv{\DefLocal{\GenericLoc}{x}{C_{1,1}}{C_{2,1}}}{\DefLocal{\GenericLoc}{x}{C_{1,2}}{C_{2,2}}}} \and
  \infer[FunLocalCong]{\ChorEquiv{C_1}{C_2}}{\ChorEquiv{\FunLocal{\GenericLoc}{F}{x}{C_1}}{\FunLocal{\GenericLoc}{F}{x}{C_2}}} \and
  \infer[FunGlobalCong]{\ChorEquiv{C_1}{C_2}}{\ChorEquiv{\FunGlobal{F}{X}{C_1}}{\FunGlobal{F}{X}{C_2}}} \and
  \infer[AppLocalCong]{\ChorEquiv{C_1}{C_2}}{\ChorEquiv{\AppLocal{\GenericLoc}{C_1}{e}}{\AppLocal{\GenericLoc}{C_2}{e}}} \and
  \infer[AppGlobalCong]{\ChorEquiv{C_{1,1}}{C_{1,2}}\\ \ChorEquiv{C_{2,1}}{C_{2,2}}}{\ChorEquiv{\AppGlobal{C_{1,1}}{C_{2,1}}}{\AppGlobal{C_{1,2}}{C_{2,2}}}}\\
  \infer[SwapSendSend]{\GenericLocI \neq \GenericLocIII\\ \GenericLocII \neq \GenericLocIII \\ \GenericLocI \neq \GenericLocIV\\ \GenericLocII \neq \GenericLocIV}{
    \ChorEquiv{\Send*{\GenericLocI}{e_1}{\GenericLocII}{x}{\Send*{\GenericLocIII}{e_2}{\GenericLocIV}{y}{C}}}
    {\Send*{\GenericLocIII}{e_2}{\GenericLocIV}{y}{\Send*{\GenericLocI}{e_1}{\GenericLocII}{x}{C}}}
  } \\
  \infer[SwapSendSync]{\GenericLocI \neq \GenericLocIII\\ \GenericLocII \neq \GenericLocIII \\ \GenericLocI \neq \GenericLocIV\\ \GenericLocII \neq \GenericLocIV}{
        \ChorEquiv{\Send*{\GenericLocI}{e}{\GenericLocII}{x}{\Sync*{\GenericLocIII}{\GenChoice}{\GenericLocIV}{C}}}
    {\Sync*{\GenericLocIII}{\GenChoice}{\GenericLocIV}{\Send*{\GenericLocI}{e}{\GenericLocII}{x}{C}}}
  }\and
  \infer[SwapSyncSend]{\GenericLocI \neq \GenericLocIII\\ \GenericLocII \neq \GenericLocIII \\ \GenericLocI \neq \GenericLocIV\\ \GenericLocII \neq \GenericLocIV}{
        \ChorEquiv{\Sync*{\GenericLocI}{\GenChoice}{\GenericLocII}{\Send*{\GenericLocIII}{e}{\GenericLocIV}{x}{C}}}
    {\Send*{\GenericLocIII}{e}{\GenericLocIV}{x}{\Sync*{\GenericLocI}{\GenChoice}{\GenericLocII}{C}}}
  }\\
  \infer[SwapSendIf]{\GenericLocI \neq \GenericLocIII\\ \GenericLocII \neq \GenericLocIII}{
    \ChorEquiv{\Send*{\GenericLocI}{e_1}{\GenericLocII}{x}{\ChorIf*{\GenericLocIII}{e_2}{C_1}{C_2}}}
    {\ChorIf*{\GenericLocIII}{e_2}{\Send*{\GenericLocI}{e_1}{\GenericLocII}{x}{C_1}}{\Send*{\GenericLocI}{e_1}{\GenericLocII}{x}{C_2}}}
  }\and
  \infer[SwapIfSend]{\GenericLocI \neq \GenericLocII\\ \GenericLocI \neq \GenericLocIII}{
    \ChorEquiv{\ChorIf*{\GenericLocI}{e_1}{\Send*{\GenericLocII}{e_2}{\GenericLocIII}{x}{C_1}}{\Send*{\GenericLocII}{e_2}{\GenericLocIII}{x}{C_2}}}
    {\Send*{\GenericLocII}{e_2}{\GenericLocIII}{x}{\ChorIf*{\GenericLocI}{e}{C_1}{C_2}}}
  }\\
  \infer[SwapSyncSync]{\GenericLocI \neq \GenericLocIII\\ \GenericLocII \neq \GenericLocIII \\ \GenericLocI \neq \GenericLocIV\\ \GenericLocII \neq \GenericLocIV}{
    \ChorEquiv{\Sync*{\GenericLocI}{\GenChoice}{\GenericLocII}{\Sync*{\GenericLocIII}{\GenChoiceII}{\GenericLocIV}{C}}}
    {\Sync*{\GenericLocIII}{\GenChoiceII}{\GenericLocIV}{\Sync*{\GenericLocI}{\GenChoice}{\GenericLocII}{C}}}
  }\\
  \infer[SwapSyncIf]{\GenericLocI \neq \GenericLocIII\\ \GenericLocII \neq \GenericLocIII}{
    \ChorEquiv{\Sync*{\GenericLocI}{\GenChoice}{\GenericLocII}{\ChorIf*{\GenericLocIII}{e}{C_1}{C_2}}}
    {\ChorIf*{\GenericLocIII}{e}{\Sync*{\GenericLocI}{\GenChoice}{\GenericLocII}{C_1}}{\Sync*{\GenericLocI}{\GenChoice}{\GenericLocII}{C_2}}}
  }\and
  \infer[SwapIfSync]{\GenericLocI \neq \GenericLocII\\ \GenericLocI \neq \GenericLocIII}{
    \ChorEquiv{\ChorIf*{\GenericLocI}{e}{\Sync*{\GenericLocII}{\GenChoice}{\GenericLocIII}{C_1}}{\Sync*{\GenericLocII}{\GenChoice}{\GenericLocIII}{C_2}}}
    {\Sync*{\GenericLocII}{\GenChoice}{\GenericLocIII}{\ChorIf*{\GenericLocI}{e}{C_1}{C_2}}}
  }\\
  \infer[SwapIfIf]{\GenericLocI \neq \GenericLocII}{
    \ChorEquiv{\ChorIf*{\GenericLocI}{e_1}{\ChorIf*{\GenericLocII}{e_2}{C_1}{C_2}}{\ChorIf*{\GenericLocII}{e_2}{C_3}{C_4}}}
    {\ChorIf*{\GenericLocII}{e_2}{\ChorIf*{\GenericLocI}{e_1}{C_1}{C_3}}{\ChorIf*{\GenericLocI}{e_1}{C_2}{C_4}}}
  }  
\end{mathparpagebreakable}


%% file: cont_sem.tex
\section{Control Language Operational Semantics}
\label{sec:contr-lang-oper}

\begin{syntax}
  \category[Label]{l} \alternative{\TauLabel}
  \alternative{\SendLabel{v}{\GenericLoc}}
  \alternative{\RecvLabel{\GenericLoc}{v}{x}}
  \alternative{\ChooseLabel{\GenChoice}{\GenericLoc}}
  \alternative{\AllowChoiceLabel{\GenericLoc}{\GenChoice}}
  \alternative{\SyncTauLabel}
  \alternative{\FunLabel(l)}
  \alternative{\ArgLabel(l)}
\end{syntax}

\begin{mathparpagebreakable}
  \infer[RetE]{
    \ExprStepsTo{e_1}{e_2}
  }{
    \ContStepsTo{\TauLabel}{\Ret(e_1)}{\Ret(e_2)}
  } \and
  \infer[IfE]{
    \ExprStepsTo{e_1}{e_2}
  }{
    \ContStepsTo{\TauLabel}{\ContIf{e_1}{E_1}{E_2}}{\ContIf{e_2}{E_1}{E_2}}
  } \and
  \infer[IfTrue]{ }{
    \ContStepsTo{\TauLabel}{\ContIf{\True}{E_1}{E_2}}{E_1}
  } \and
  \infer[IfFalse]{ }{
    \ContStepsTo{\TauLabel}{\ContIf{\False}{E_1}{E_2}}{E_2}
  } \and
  \infer[SendE]{
    \ExprStepsTo{e_1}{e_2}
  }{
    \ContStepsTo{\TauLabel}{\ContSend{e_1}{\GenericLoc}{E}}{\ContSend{e_2}{\GenericLoc}{E}}
  }\and
  \infer[SendV]{
    \ExprValue{v}
  }{
    \ContStepsTo{\SendLabel{v}{\GenericLoc}}{\ContSend{v}{\GenericLoc}{E}}{E}
  } \and
  \infer[RecvV]{
    \ExprValue{v}
  }{
    \ContStepsTo{\RecvLabel{\GenericLoc}{v}{x}}{\ContReceive{x}{\GenericLoc}{E}}{\Subst{E}{\SingleSubst{x}{v}}}
  } \and
  \infer[Choose]{ }{
    \ContStepsTo{\ChooseLabel{\GenChoice}{\GenericLoc}}{\ContChoose{\GenChoice}{\GenericLoc}{E}}{E}
  } \and
  \infer[AllowChoiceL]{ }{
    \ContStepsTo{\AllowChoiceLabel{\GenericLoc}{\LChoice}}
    {\AllowChoiceLR{\GenericLoc}{E_1}{E_2}}{E_1}
  } \and
  \infer[AllowChoiceR]{ }{
    \ContStepsTo{\AllowChoiceLabel{\GenericLoc}{\RChoice}}
    {\AllowChoiceLR{\GenericLoc}{E_1}{E_2}}{E_2}
  } \and
  \infer[LetRetArg]{
    \ContStepsTo{l}{E_1}{E_1'}
  }{
    \ContStepsTo{\ArgLabel(l)}{\LetRet{x}{E_1}{E_2}}{\LetRet{x}{E_1'}{E_2}}
  } \and
  \infer[LetRet]{
    \ExprValue{v}
  }{
    \ContStepsTo{\SyncTauLabel}{\LetRet{x}{\Ret(v)}{E_2}}{\Subst{E_2}{\SingleSubst{x}{v}}}
  } \and
  \infer[AppLocalFun]{
    \ContStepsTo{l}{F_1}{F_2}
  }{
    \ContStepsTo{\FunLabel(l)}{\ContAppLocal{F_1}{e}}{\ContAppLocal{F_2}{e}}
  } \and
  \infer[AppLocalArg]{
    \ExprStepsTo{e_1}{e_2}
  }{
    \ContStepsTo{\TauLabel}{\ContAppLocal{F}{e_1}}{\ContAppLocal{F}{e_2}}
  } \and
  \infer[AppLocal]{
    \ExprValue{v}
  }{
    \ContStepsTo{\SyncTauLabel}{\ContAppLocal{\left(\ContFunLocal{F}{x}{E}\right)}{v}}{\Subst{E}{\SingleSubst{x}{v}}}
  }
  \infer[AppGlobalFun]{
    \ContStepsTo{l}{F_1}{F_2}
  }{
    \ContStepsTo{\FunLabel(l)}{\ContAppGlobal{F_1}{A}}{\ContAppGlobal{F_2}{A}}
  } \and
  \infer[AppGlobalArg]{
    \ContStepsTo{l}{A_1}{A_2}
  }{
    \ContStepsTo{\ArgLabel(l)}{\ContAppGlobal{F}{A_1}}{\ContAppGlobal{F}{A_2}}
  } \and
  \infer[AppGlobal]{
    \ContValue(V)
  }{
    \ContStepsTo{\SyncTauLabel}{\ContAppGlobal{\left(\ContFunGlobal{F}{X}{E}\right)}{V}}{\Subst{E}{\SingleSubst{X}{V}}}
  }
\end{mathparpagebreakable}


%% file: cont_merge.tex
\section{Control Program Merge Definition}
\label{sec:contr-progr-merge}

If there is no pattern below such that $\ContMerge{E_1}{E_2} = E_3$, then $\Undefed{(\ContMerge{E_1}{E_2})}$.

\begin{mathparpagebreakable}
  \ContMerge{X}{X} \triangleq X \and
  \ContMerge{\Unit}{\Unit} \triangleq \Unit \and
  \ContMerge{\Ret(e)}{\Ret(e)} \triangleq \Ret(e) \and
  \ContMerge{(\ContIf{e}{E_{1,1}}{E_{1,2}})}{(\ContIf{e}{E_{2,1}}{E_{2,2}})} \triangleq \ContIf{e}{\ContMerge{E_{1,1}}{E_{2,1}}}{\ContMerge{E_{1,2}}{E_{2,2}}} \and
  \ContMerge{(\ContSend{e}{\GenericLoc}{E_1})}{(\ContSend{e}{\GenericLoc}{E_2})} \triangleq \ContSend{e}{\GenericLoc}{\ContMerge{E_1}{E_2}} \and
  \ContMerge{(\ContReceive{x}{\GenericLoc}{E_1})}{(\ContReceive{x}{\GenericLoc}{E_2})} \triangleq \ContReceive{x}{\GenericLoc}{\ContMerge{E_1}{E_2}} \and
  \ContMerge{(\ContChoose{\GenChoice}{\GenericLoc}{E_1})}{(\ContChoose{\GenChoice}{\GenericLoc}{E_2})} \triangleq \ContChoose{\GenChoice}{\GenericLoc}{\ContMerge{E_1}{E_2}} \and
  \ContMerge{\left(\AllowChoiceL*{\GenericLoc}{E_1}\right)}{\left(\AllowChoiceL*{\GenericLoc}{E_2}\right)} \triangleq \AllowChoiceL*{\GenericLoc}{\ContMerge{E_1}{E_2}} \and
  \ContMerge{\left(\AllowChoiceL*{\GenericLoc}{E_1}\right)}{\left(\AllowChoiceR*{\GenericLoc}{E_2}\right)} \triangleq \AllowChoiceLR*{\GenericLoc}{E_1}{E_2} \and
  \ContMerge{\left(\AllowChoiceL*{\GenericLoc}{E_1}\right)}{\left(\AllowChoiceLR*{\GenericLoc}{E_{2,1}}{E_{2,2}}\right)} \triangleq \AllowChoiceLR*{\GenericLoc}{\ContMerge{E_1}{E_{2,1}}}{E_{2,2}} \and
  \ContMerge{\left(\AllowChoiceR*{\GenericLoc}{E_1}\right)}{\left(\AllowChoiceL*{\GenericLoc}{E_2}\right)} \triangleq \AllowChoiceLR*{\GenericLoc}{E_2}{E_1} \and
  \ContMerge{\left(\AllowChoiceR*{\GenericLoc}{E_1}\right)}{\left(\AllowChoiceR*{\GenericLoc}{E_2}\right)} \triangleq \AllowChoiceR*{\GenericLoc}{\ContMerge{E_1}{E_2}} \and
  \ContMerge{\left(\AllowChoiceR*{\GenericLoc}{E_1}\right)}{\left(\AllowChoiceLR*{\GenericLoc}{E_{2,1}}{E_{2,2}}\right)} \triangleq \AllowChoiceLR*{\GenericLoc}{E_{2,1}}{\ContMerge{E_1}{E_{2,2}}} \and
  \ContMerge{\left(\AllowChoiceLR*{\GenericLoc}{E_{1,1}}{E_{1,2}}\right)}{\left(\AllowChoiceL*{\GenericLoc}{E_2}\right)} \triangleq
    \AllowChoiceLR*{\GenericLoc}{\ContMerge{E_{1,1}}{E_2}}{E_{1,2}} \and
  \ContMerge{\left(\AllowChoiceLR*{\GenericLoc}{E_{1,1}}{E_{1,2}}\right)}{\left(\AllowChoiceR*{\GenericLoc}{E_2}\right)} \triangleq
    \AllowChoiceLR*{\GenericLoc}{E_{1,1}}{\ContMerge{E_{1,2}}{E_2}} \and
  \ContMerge{\left(\AllowChoiceLR*{\GenericLoc}{E_{1,1}}{E_{1,2}}\right)}{\left(\AllowChoiceLR*{\GenericLoc}{E_{2,1}}{E_{2,2}}\right)} \triangleq
    \AllowChoiceLR*{\GenericLoc}{\ContMerge{E_{1,1}}{E_{2,1}}}{\ContMerge{E_{1,2}}{E_{2,2}}} \and
  \ContMerge{(\LetRet{x}{E_{1,1}}{E_{1,2}})}{(\LetRet{x}{E_{2,1}}{E_{2,2}})} \triangleq \LetRet{x}{\ContMerge{E_{1,1}}{E_{2,1}}}{\ContMerge{E_{1,2}}{E_{2,2}}} \and
  \ContMerge{(\FunLocal{\GenericLoc}{F}{x}{E})}{(\FunLocal{\GenericLoc}{F}{x}{E})} \triangleq \FunLocal{\GenericLoc}{F}{x}{E} \and
  \ContMerge{(\FunGlobal{F}{X}{E})}{(\FunGlobal{F}{X}{E})} \triangleq \FunGlobal{F}{X}{E} \\
  \ContMerge{(\AppLocal{\GenericLoc}{E_1}{e})}{(\AppLocal{\GenericLoc}{E_2}{e})} \triangleq \AppLocal{\GenericLoc}{\ContMerge{E_1}{E_2}}{e} \and
  \ContMerge{(\AppGlobal{E_{1,1}}{E_{1,2}})}{(\AppGlobal{E_{2,1}}{E_{2,2}})} \triangleq \AppGlobal{(\ContMerge{E_{1,1}}{E_{2,1}})}{(\ContMerge{E_{1,2}}{E_{2,2}})}
\end{mathparpagebreakable}


%% file: lnd.tex
\section{The Less-Nondeterminism Relation}
\label{sec:less-nond-relat}

\begin{mathparpagebreakable}
  \infer{ }{\LND{X}{X}} \and
  \infer{ }{\LND{\Unit}{\Unit}} \and
  \infer{ }{\LND{\Ret(e)}{\Ret(e)}} \and
  \infer{\LND{E_{1,1}}{E_{2,1}}\\ \LND{E_{1,2}}{E_{2,2}}}
  {\LND{\ContIf{e}{E_{1,1}}{E_{1,2}}}{\ContIf{e}{E_{2,1}}{E_{2,2}}}} \and
  \infer{\LND{E_1}{E_2}}
  {\LND{\ContSend{e}{\GenericLoc}{E_1}}{\ContSend{e}{\GenericLoc}{E_2}}} \and
  \infer{\LND{E_1}{E_2}}
  {\LND{\ContReceive{x}{\GenericLoc}{E_1}}{\ContReceive{x}{\GenericLoc}{E_2}}} \and
  \infer{\LND{E_1}{E_2}}
  {\LND{\ContChoose{\GenChoice}{\GenericLoc}{E_1}}{\ContChoose{\GenChoice}{\GenericLoc}{E_2}}} \and
  \infer{\LND{E_1}{E_2}}
  {\LND{\AllowChoiceL{\GenericLoc}{E_1}}{\AllowChoiceL{\GenericLoc}{E_{2,1}}}} \and
  \infer{\LND{E_1}{E_{2,1}}}
  {\LND{\AllowChoiceL{\GenericLoc}{E_1}}{\AllowChoiceLR{\GenericLoc}{E_{2,1}}{E_{2,2}}}} \and
  \infer{\LND{E_1}{E_2}}
  {\LND{\AllowChoiceR{\GenericLoc}{E_1}}{\AllowChoiceR{\GenericLoc}{E_{2,1}}}} \and
  \infer{\LND{E_1}{E_{2,2}}}
  {\LND{\AllowChoiceR{\GenericLoc}{E_1}}{\AllowChoiceLR{\GenericLoc}{E_{2,1}}{E_{2,2}}}} \and
  \infer{\LND{E_{1,1}}{E_{2,1}}\\\LND{E_{1,2}}{E_{2,2}}}
  {\LND{\AllowChoiceLR{\GenericLoc}{E_{1,1}}{E_{1,2}}}
    {\AllowChoiceLR{\GenericLoc}{E_{2,1}}{E_{2,2}}}} \and
  \infer{\LND{E_{1,1}}{E_{2,1}}\\\LND{E_{1,2}}{E_{2,2}}}
  {\LND{\LetRet{x}{E_{1,1}}{E_{1,2}}}{\LetRet{x}{E_{2,1}}{E_{2,2}}}}\and
  \infer{ }{\LND{\ContFunLocal{F}{x}{E}}{\ContFunLocal{F}{x}{E}}} \and
  \infer{\LND{E_1}{E_2}}{\LND{\ContAppLocal{E_1}{e}}{\ContAppLocal{E_2}{e}}} \and
  \infer{ }{\LND{\ContFunGlobal{F}{X}{E}}{\ContFunGlobal{F}{X}{E}}} \and
  \infer{\LND{E_{1,1}}{E_{2,1}}\\ \LND{E_{1,2}}{E_{2,2}}}
  {\LND{\ContAppGlobal{E_{1,1}}{E_{1,2}}}{\ContAppGlobal{E_{2,1}}{E_{2,2}}}}
\end{mathparpagebreakable}


%% file: thms.tex
\section{Additional Theorems}
\label{sec:additional-theorems}

\subsection{Additional Properties of the Pirouette Type System}
\label{sec:addit-prop-piro}

While we do not explicitly supply the structural rules, our type system admits them:
\begin{thm}[Pirouette Types Structural Rules]
  The choreographic type system admits weakening, exchange, and strengthening in both contexts.
  In other words, the following rules are admissible:
  \begin{mathpar}
    \infer[Local Exchange]{\ChorTyping{\Gamma, \ChorTypingTriple{\GenericLocI}{x}{\ExprTypeTauI}, \ChorTypingTriple{\GenericLocII}{y}{\ExprTypeTauII}, \Gamma'}{\Delta}{C}{\ChorTypeTauIII}}
    {\ChorTyping{\Gamma, \ChorTypingTriple{\GenericLocII}{y}{\ExprTypeTauII}, \ChorTypingTriple{\GenericLocI}{x}{\ExprTypeTauI}, \Gamma'}{\Delta}{C}{\ChorTypeTauIII}} \and
    \infer[Global Exchange]{\ChorTyping{\Gamma}{\Delta, \ChorTypingPair{X}{\ChorTypeTauI}, \ChorTypingPair{Y}{\ChorTypeTauII}, \Delta'}{C}{\ChorTypeTauIII}}
    {\ChorTyping{\Gamma}{\Delta, \ChorTypingPair{Y}{\ChorTypeTauII}, \ChorTypingPair{X}{\ChorTypeTauI}, \Delta'}{C}{\ChorTypeTauIII}} \\
    \infer[Local Weakening]{\ChorTyping{\Gamma}{\Delta}{C}{\ChorTypeTauI}}{\ChorTyping{\Gamma, \ChorTypingTriple{\GenericLoc}{x}{\ExprTypeTauII}}{\Delta}{C}{\ChorTypeTauI}} \and
    \infer[Global Weakening]{\ChorTyping{\Gamma}{\Delta}{C}{\ChorTypeTauI}}{\ChorTyping{\Gamma}{\Delta, \ChorTypingPair{X}{\ChorTypeTauII}}{C}{\ChorTypeTauI}} \\
    \infer[Local Strengthening]{\ChorTyping{\Gamma, \ChorTypingTriple{\GenericLoc}{x}{\ExprTypeTauI}}{\Delta}{C}{\ChorTypeTauII}\\ x \notin \FEV[\GenericLoc](C)}{\ChorTyping{\Gamma}{\Delta}{C}{\ChorTypeTauII}} \and
    \infer[Global Strengthening]{\ChorTyping{\Gamma}{\Delta, \ChorTypingPair{X}{\ChorTypeTauI}}{C}{\ChorTypeTauII}\\ X \notin \FCV(C)}{\ChorTyping{\Gamma}{\Delta}{C}{\ChorTypeTauII}}
  \end{mathpar}
\end{thm}

\subsection{Additional Properties of Choreography Equivalence}
\label{sec:addit-prop-chor}

Our syntactic operations respect choreography equivalence:
\begin{thm}[Equivalence of Syntactic Operations]
  \label{thm:equiv-syntax}
  If $\ChorEquiv{C_1}{C_1'}$ and $\ChorEquiv{C_2}{C_2'}$, then all of the following statements are true:
  \begin{itemize}
  \item $\FEV(C_1) = \FEV(C_1')$
  \item $FCV(C_1) = \FCV(C_1')$
  \item $\LocationNames(C_1) = \LocationNames(C_1')$
  \item if $C_1$ is a value, then so is $C_1'$
  \item $\ChorEquiv{\Subst{C_1}{\SingleChorExprSubst{\GenericLoc}{x}{e}}}{\Subst{C_1'}{\SingleChorExprSubst{\GenericLoc}{x}{e}}}$
  \item $\ChorEquiv{\Subst{C_1}{\SingleSubst{X}{C_2}}}{\Subst{C_1'}{\SingleSubst{X}{C_2'}}}$
  \end{itemize}
\end{thm}

\subsection{Additional Properties of Merging}
\label{sec:addit-prop-merg}
\begin{thm}
  The following are true for all $E$, $E_1$, $E_2$, and $E_3$:
  \begin{itemize}
  \item $\ContMerge{E}{E} = E$
  \item $\ContMerge{E_1}{E_2} = \ContMerge{E_2}{E_1}$
  \item $\ContMerge{(\ContMerge{E_1}{E_2})}{E_3} = \ContMerge{E_1}{(\ContMerge{E_2}{E_3})}$
  \end{itemize}
\end{thm}

\begin{thm}
  The following are true for all control-language expressions~$E$, $E_1$, and $E_2$, along with all control-language \emph{values} $V$.
  \begin{itemize}
  \item $\FV(\ContMerge{E_1}{E_2}) = \FV(E_1) \cup \FV(E_2)$
  \item $\FEV(\ContMerge{E_1}{E_2}) = \FEV(E_1) \cup \FEV(E_2)$
  \item $\ContMerge{(\Subst{E_1}{\SingleSubst{x}{e}})}{(\Subst{E_2}{\SingleSubst{x}{e}})} = \Subst{(\ContMerge{E_1}{E_2})}{\SingleSubst{x}{e}}$
  \item $\ContMerge{(\Subst{E_1}{\SingleSubst{X}{E}})}{(\Subst{E_2}{\SingleSubst{X}{E}})} = \Subst{(\ContMerge{E_1}{E_2})}{\SingleSubst{X}{E}}$
  \item Either $\Undefed{(\ContMerge{V}{E})}$ or $E = V = \ContMerge{V}{E} = \ContMerge{E}{V}$
  \item If $\ContMerge{E_1}{E_2} = V$, then $E_1 = E_2 = V$ 
  \end{itemize}
\end{thm}

\begin{thm}
  If $\ContStepsTo{L}{E_1}{E_1'}$, $\ContStepsTo{L}{E_2}{E_2'}$, and $\Defed{(\ContMerge{E_1}{E_2})}$, then $\Defed{(\ContMerge{E_1'}{E_2'})}$ and $\ContStepsTo{L}{\ContMerge{E_1}{E_2}}{\ContMerge{E_1'}{E_2'}}$.
\end{thm}

\subsection{Additional Properties of Endpoint Projection}
\label{sec:addit-prop-endp}

Our second property tells us that projection plays nicely with all of the syntactic operations that we have defined on choreographies:
\begin{thm}
  \label{thm:projection-syntax-ops}
  All of the following are true:
  \begin{itemize}
  \item $\FEV[\GenericLoc](C) = \FEV(\ProjectC{C}{\GenericLoc})$
  \item $\FCV(C) = \FV(\ProjectC{C}{\GenericLoc})$
  \item $\ProjectC{\Subst{C_1}{\SingleSubst{X}{C_2}}}{\GenericLoc} = \Subst{\ProjectC{C_1}{\GenericLoc}}{\SingleSubst{X}{\ProjectC{C_2}{\GenericLoc}}}$
  \item $\ProjectC{\Subst{C}{\SingleChorExprSubst{\GenericLocI}{x}{e}}}{\GenericLocII} =
    \left\{
      \begin{array}{ll}
        \Subst{\ProjectC{C}{\GenericLocI}}{\SingleSubst{x}{e}} & \text{if}~\GenericLocI = \GenericLocII\\
        \ProjectC{C}{\GenericLocI} & \text{otherwise}
      \end{array}
      \right.$
    \item If $V$ is a value, then so is $\ProjectC{V}{\GenericLoc}$.
    \item If $\ProjectC{C}{\GenericLoc}$ is a value for every $\GenericLoc \in \mathcal{L}$, then $C$ is a value.
  \end{itemize}
\end{thm}

\subsection{Additional Properties of \LNDRel}
\label{sec:addit-prop-lndr}

\begin{thm}[\LNDRel~is a Partial Order]
  The following are all true for all $E$, $E_1$, $E_2$, and $E_3$.
  \begin{itemize}
  \item $\LND{E}{E}$
  \item if $\LND{E_1}{E_2}$ and $\LND{E_2}{E_1}$, then $E_1 = E_2$
  \item if $\LND{E_1}{E_2}$ and $\LND{E_2}{E_3}$, then $\LND{E_1}{E_3}$
  \end{itemize}
\end{thm}

\begin{thm}
  The following are all true for all choices of $E_i$.
  \begin{itemize}
  \item If $\LND{E_1}{E_2}$, then $\LND{(\Subst{E_1}{\SingleSubst{x}{e}})}{(\Subst{E_2}{\SingleSubst{x}{e}})}$
  \item If $\LND{E_1}{E_2}$, then $\LND{(\Subst{E_1}{\SingleSubst{X}{E_3}})}{(\Subst{E_2}{\SingleSubst{X}{E_3}})}$
  \item $\LND{E_1}{\ContMerge{E_1}{E_2}}$ and $\LND{E_2}{\ContMerge{E_1}{E_2}}$
  \item If $\LND{E_1}{E_2}$ and $\LND{E_3}{E_4}$ and $\Defed{(\ContMerge{E_1}{E_3})}$ then $\Defed{(\ContMerge{E_2}{E_4})}$ and $\LND{\ContMerge{E_1}{E_3}}{\ContMerge{E_2}{E_4}}$
  \item If $\LND{E_1}{E_2}$ and either of $E_1$ or $E_2$ is a value, then $E_1 = E_2$.
  \end{itemize}
\end{thm}

\begin{thm}[Steps Lift Across~\LNDRel]
  \label{thm:lift-local}
  If $\LND{E_1}{E_2}$ and $\ContStepsTo{l}{E_1}{E_1'}$, then there is a $E_2'$ such that $\LND{E_1'}{E_2'}$ and $\ContStepsTo{l}{E_2}{E_2'}$.
\end{thm}

When a more-nondeterministic expression takes a step, it may be taking advantage of the extra nondeterminism to take a step not available to the original.
We can take advantage of the fact that our operational semantics is defined via a labeled-transition system to lift steps which do not resolve nondeterminism.
In order to turn this into a theorem, we need to determine when a label may reduce an external choice.
We do this with the predicate $\IsChoiceLabel{\GenChoice}(l)$, which we define as follows:
\begin{mathpar}
  \infer{ }{\IsChoiceLabel{\GenChoice}(\AllowChoiceLabel{\GenericLoc}{\GenChoice})} \and
  \infer{\IsChoiceLabel{\GenChoice}(l)}{\IsChoiceLabel{\GenChoice}(\FunLabel(l))} \and
  \infer{\IsChoiceLabel{\GenChoice}(l)}{\IsChoiceLabel{\GenChoice}(\ArgLabel(l))} \and
\end{mathpar}

With the definition of $\IsChoiceLabel{\GenChoice}(-)$ in hand, we can prove the following theorem:
\begin{thm}[Lowering Steps Across~\LNDRel]
  \label{thm:lower-local}
  If $\LND{E_1}{E_2}$ and $\ContStepsTo{l}{E_2}{E_2'}$ where $\lnot\IsChoiceLabel{d}(l)$, then there is a $E_1'$ such that $\LND{E_1'}{E_2'}$ and $\ContStepsTo{l}{E_1}{E_1'}$.
\end{thm}
